\documentclass[11pt, letterpaper]{article}
\usepackage{fullpage}
\usepackage[T1]{fontenc}

\usepackage[bookmarks=true,colorlinks=true,pdfpagemode=none,linkcolor=blue,citecolor=blue]{hyperref}
\usepackage{amsmath,amsfonts,amsthm, amssymb}
\usepackage{algorithm, algorithmicx}
\usepackage[compatible]{algpseudocode}
\usepackage{xcolor, color, multirow, enumerate, graphicx}

\usepackage{thmtools}
\usepackage{thm-restate}

\newtheorem{thm}{Theorem}[section]
\newtheorem{corollary}[thm]{Corollary}
\newtheorem{lem}[thm]{Lemma}
\newtheorem{lemma}[thm]{Lemma}

\newtheorem{defn}[thm]{Definition}

\newtheorem{fact}[thm]{Fact}
\newtheorem{assumption}[thm]{Assumption}
\newtheorem{remark}[thm]{Remark}

\def\poly{\operatorname{poly}}
\def\polylog{\operatorname{polylog}}
\global\long\def\cut{c^+}

\newcommand{\oyy}{{{\bar{w}}}} %
\newcommand{\hF}{\widehat{F}}
\newcommand{\hy}{\widehat{w}}
\newcommand{\hyy}{\widehat{w}}
\newcommand{\tx}{\widetilde{x}}
\newcommand{\tO}[1]{\widetilde{O}\left(#1\right)}

\newcommand{\txx}{{\widetilde{{x}}}} %
\newcommand{\yy}{{{w}}} %
\newcommand{\tyy}{{\widetilde{{w}}}}%
\newcommand{\ty}{{\widetilde{{w}}}} %
\newcommand{\zz}{{{z}}} %
\newcommand{\zerov}{{0}} %
\newcommand{\ww}{{\mathit{w}}} %
\newcommand{\xx}{{{x}}} %

\newcommand{\eps}{\varepsilon}
\newcommand{\m}{r}
\newcommand{\C}{\alpha}

\newcommand{\smo}{L}
\newcommand{\str}{\sigma}
\newcommand{\fmax}{F_{\max}}

\newcommand{\evo}{\textnormal{EO}} \newcommand{\opto}{\mathcal{O}} \newcommand{\tmf}{T_{\mathrm{maxflow}}}

\newcommand{\lmin}{\lambda_{\mathrm{min}}}
\newcommand{\lmax}{\lambda_{\mathrm{max}}}
\def\ln{\log}
\usepackage[normalem]{ulem}

\newif\ifdebug

\debugtrue

\ifdebug
\newcommand\redout{\bgroup\markoverwith{\textcolor{red}{\rule[0.5ex]{2pt}{0.8pt}}}\ULon}

\else
\newcommand\redout{}

\fi

\renewcommand{\algorithmiccomment}[1]{\hfill\textcolor{gray}{// #1}}
\newcommand{\RETURN}{\STATE {\bf return} }

\algdef{SE}[SUBALG]{Indent}{EndIndent}{}{\algorithmicend\ }%
\algtext*{Indent}
\algtext*{EndIndent}

\begin{document}
\clubpenalty=10000
\widowpenalty = 10000

\title{Decomposable Submodular Function Minimization\\ via Maximum Flow}
\author{
Kyriakos Axiotis\thanks{MIT, \tt{kaxiotis@mit.edu}}
\and
Adam Karczmarz\thanks{University of Warsaw, \tt{a.karczmarz@mimuw.edu.pl}}
\and
Anish Mukherjee\thanks{University of Warsaw, \tt{anish343@gmail.com}} 
\and 
Piotr Sankowski\thanks{University of Warsaw, \tt{sank@mimuw.edu.pl}}
\and
Adrian Vladu\thanks{CNRS \& IRIF, Universit\'{e} de Paris, \tt{vladu@irif.fr}}
}

\date{}
\maketitle

\begin{abstract}
This paper bridges discrete and continuous optimization approaches for decomposable submodular function minimization, in both the standard and parametric settings. 

We provide improved running times for this problem by reducing it to a number of calls to a maximum flow oracle. When each function in the decomposition acts on $O(1)$ elements of the ground set $V$ and is polynomially bounded, our running time is up to polylogarithmic factors equal to that of solving maximum flow in a sparse graph with $O(\vert V \vert)$ vertices and polynomial integral capacities.

We achieve this by providing a simple iterative method which can optimize to high precision any convex function defined on the submodular base polytope, provided we can efficiently minimize it on the base polytope corresponding to the cut function of a certain graph that we construct. We solve this minimization problem by lifting the solutions of a parametric cut problem, which we obtain via a new efficient combinatorial reduction to maximum flow. 
This reduction is of independent interest and implies some previously unknown bounds for the parametric minimum $s,t$-cut problem in multiple settings.
\end{abstract}

\section{Introduction}
A significant amount of work has been dedicated to the study of submodular functions. 
While this topic has garnered a lot of excitement from the theory community due to its the multiple connections to diverse algorithmic areas~\cite{lovasz1983submodular, grotschel1981ellipsoid}, on the practical side minimizing submodular functions has been intensively used to model discrete problems in machine learning. MAP inference in Markov Random Fields~\cite{kohli2009robust}, image segmentation~\cite{arora2012generic, shanu2016min}, clustering~\cite{narasimhan2007local}, corpus extraction problems~\cite{lin2011optimal} are just a few success stories of submodular minimization.

Polynomial time algorithms for this problem have been known ever since the 80's~\cite{grotschel1981ellipsoid}, and they have seen major running time improvements in more recent years~\cite{schrijver2000combinatorial, iwata2003faster, fleischer2003push, orlin2009faster, lee2015faster, chakrabarty2017subquadratic}. However, the massive scale of the problems that use submodular minimization nowadays drives the need for further developments.

One great advantage offered by the submodular functions that occur in practice is that they are \emph{structured}. For example, in many common cases (hypergraph cuts~\cite{veldt2020minimizing}, covering functions~\cite{stobbe2010efficient}, MAP inference~\cite{fix2013structured, kohli2009robust, vicente2009joint}) these can be decomposed into sums of simple submodular functions defined on small subsets. For these instances, prior work~\cite{jegelka2013reflection, nishihara2014convergence, ene2015random, ene2017decomposable, kumar2019fast} has focused on providing efficient algorithms in the regime where the functions in the decomposition admit fast optimization oracles.

Notably, many of these recent developments have leveraged a mix of ideas coming from both discrete and continuous optimization. In particular, Ene et al.~\cite{ene2017decomposable} present algorithms for decomposable function minimization that are based on both continuous methods (i.e. gradient descent) and discrete algorithms, as the authors employ a version of the preflow-push algorithm for maximum flow~\cite{goldberg1988new}. As this work was paralleled by multiple improvements to the running time for maximum flow~\cite{madry2013navigating, madry2016computing, liu2020faster, kathuria2020unit, brand2021minimum, flow-sparse}, most of which stemmed from innovations in convex optimization, it seemed plausible that the same new optimization techniques could be helpful for improving the running times of other fundamental problems in combinatorial optimization, including submodular function minimization. In this context, a particularly intriguing question emerged:

\begin{center}
\emph{Can we leverage the techniques used to obtain faster algorithms for maximum flow to provide faster algorithms for submodular function minimization?}
\end{center}

We answer this question in the affirmative, by showing how to solve decomposable submodular function minimization using \emph{black-box access} to any routine that can compute the maximum flow in a capacitated directed graph. To compare the running times, in the case where all the functions in the decomposition act on $O(1)$ elements of the ground set and are polynomially bounded  (such as the case of a hypergraph cut function, with $O(1)$ sized hyperedges), our algorithm has -- up to polylogarithmic factors -- the same running time as that of computing maximum flow in a sparse graph with $O(\vert V \vert)$ vertices, and polynomial integral capacities~\cite{goldberg1998beyond, brand2021minimum, flow-sparse}.

As it turns out, to achieve this it is not sufficient to directly use off-the-shelf maximum flow algorithms. Instead, our approach is based on solving submodular minimization in the more general parametric setting, where we further parametrize the problem with an additional time-dependent penalty term on the elements in the set, and want to simultaneously solve \emph{all} the problems in this family. In turn, our reduction requires solving the parametric minimum cut problem, which has been intensely studied in the classical graph theoretic literature~\cite{GalloGT89, McCormick99, TarjanWZZM06, GranotMQT12}. In this setting, which is essentially a particular case of parametric submodular minimization, the capacities of certain arcs in the graph evolve in a monotonic fashion.

While some of the existing work on parametric cuts and flows does provide efficient algorithms via reductions to maximum flow~\cite{TarjanWZZM06}, the type of parametric capacities it supports does not cover the requirements for our more general scenario. Therefore, we develop a new efficient algorithm for computing  parametric cuts under a broad range of parametric capacities. Our algorithm is nearly optimal from the perspective of weakly-polynomial time algorithms, since its running time matches (up to polylogarithmic factors involving certain parameters) that of the fastest maximum flow algorithm in a directed graph with integer capacities. In addition, our reduction also provides novel improvements in several other regimes, involving the strongly polynomial case, and that of planar graphs, both of which may be of independent interest.

\subsection{Our Results}
In this paper we establish further connections between discrete and continuous optimization to provide an efficient algorithm for solving the decomposable submodular function minimization problem in the more general parametric setting. Our algorithm is at its core based on a continuous optimization method, but whose progress steps are driven by a new combinatorial algorithm we devise for the parametric cut problem. In this sense, our approach leverages the paradigm of combinatorial preconditioning from scientific computing literature~\cite{spielman2004smoothed, bern2006support, koutis2009combinatorial, toledo2010combinatorial}.

To properly state our main result, we need to introduce some notation. Let $V = \{1,\dots,n\}$ and let $F:2^V \rightarrow \mathbb{N}$ a submodular set function with the special property that 
\[
F(S) = \sum_{i=1}^\m F_i(S),\quad\textnormal{for all }S \subseteq V\,,
\]
where each $F_i: 2^V \rightarrow \mathbb{N}$ is a submodular set function acting on a subset $V_i \subseteq V$ of elements, in the sense that $F_i(S) = F_i(S \cap V_i)$ for all $S\subseteq V$. Let $\evo_i$ be the time required to evaluate $F_i(S)$ for any $S\subseteq V_i$, and let $\opto_i$ be the time required to minimize $F_i(S) + w(S)$ over $V_i$, where $w$ is any linear function, and suppose that $\max_{S\subseteq V}F(S) = n^{O(1)}$.
Furthermore, for each $i \in V$ let $\psi_i : \mathbb{R} \rightarrow \mathbb{R}$ be a strictly convex function satisfying $n^{-O(1)} \leq \vert \psi_i''(x) \vert \leq n^{O(1)}$ and $\vert \psi_i'(0) \vert \leq n^{O(1)}$. 

Then our main theorem is the following.
\begin{thm}\label{thm:main}
There is an algorithm which for all $\lambda \in \mathbb{R}$ simultaneously optimizes the objective
\[
\min_{S\subseteq V} F(S) + \sum_{i \in S} \psi'_i(\lambda)
\]
by returning a vector $x$ such that for any $\lambda\in \mathbb{R}$ the set
$S^\lambda = \{u\ :\ x_u \geq \lambda\}$ satisfies
\begin{align*}
F(S^\lambda) + \sum\limits_{u\in S^\lambda} \psi'(\lambda) \leq 
\underset{S\subseteq V}\min\, F(S) + \sum\limits_{u\in S} \psi'(\lambda) + \eps\,.
\end{align*}
Furthermore, if $\tmf(n,m)$ is the time required to compute the maximum flow in a directed graph with polynomially bounded integral capacities, then our algorithm runs in time
\begin{align*}
\widetilde{O}\bigg( \max_i \vert V_i \vert^2 \bigg( & \sum_{i=1}^{\m} \vert V_i \vert^2 \opto_i
+ \tmf\bigg(n, n+\sum_{i=1}^\m \vert V_i \vert^2\bigg) \bigg) \ln \frac{1}{\epsilon}\bigg)\,.
\end{align*}
\end{thm}
To better understand this result, let us consider the case where each submodular function in the decomposition acts on a small number of elements, i.e. $\vert V_i \vert = O(1)$. In this case we have the following corollary:
\begin{corollary}
If each function $F_i$ in the decomposition acts on $O(1)$ elements of the ground set, then we can the solve parametric submodular minimization problem to $\epsilon$ precision in time
\[
\widetilde{O}\left( \tmf(n, n+r)\ln \frac{1}{\epsilon}\right)\,.
\]
\end{corollary}
While our statements concern the parametric setting, it is easy to use them to recover the solution to the standard submodular minimization problem. Simply by letting $\psi'_i(t) = t$ for all $i$, and thresholding the returned vector at $0$ we obtain the desired result.
Using Goldberg-Rao~\cite{goldberg1998beyond} or the current state of the art algorithms for maximum flow~\cite{flow-dense, flow-sparse}, we see that this significantly improves over all the previous algorithms for decomposable submodular minimization, in the regime where all sets $V_i$ are small. Following~\cite{ene2017decomposable} it has remained widely open whether algorithms improving upon the $\widetilde{O}(\min\{n^2, nr\} \ln^{O(1)}(1/\epsilon))$ running time exist, and it has been conjectured that faster running times could be obtained by leveraging the newer techniques for graph algorithms based on interior point methods.

Using~\cite{flow-dense, flow-sparse}, we obtain a running time of $\widetilde{O}\left( \min\{n^{3/2}+r, (n+r)^{3/2 - 1/328} \}  \ln{1/\epsilon} \right)$.

The crucial subroutine our algorithm is based on is a novel efficient algorithm for solving the parametric cut problem using a maximum flow oracle.
We give an overview of our reduction
and its additional applications
in Section~\ref{sec:parametric-reduction},
and describe it in detail in Appendix~\ref{sec:parm-st-cut}.

\iffalse
This yields the following theorem, which we re-state in more detail in Section~\ref{sec:parm-st-cut}. We employ the definitions from standard literature, which we reiterate in Section~\ref{sec:parm-st-cut}.
\begin{thm}\label{t:aprx-pmc-intro}
Let $R=\lmax-\lmin$ be an integral multiple of $\eps>0$. And let $G_\lambda$ be the graph with capacities parametrized by $\lambda$. Then for any $\lambda=\lmin+\ell\eps$ and any integer $\ell\in [0,R/\eps]$, the $\eps$-approximate parametric min $s,t$-cut in $G$ can
be computed in $O(\tmf(n,m\log{n})\cdot\log{\frac{R}{\eps}}\cdot \log{n})$ time.
\end{thm}
\fi

\subsection{Previous Work}
\paragraph{Related Works on Submodular Minimization}
Submodular function minimization is a classical problem in combinatorial optimization, which goes back to the seminal work of Edmonds~\cite{edmonds2003submodular}. The first polynomial-time algorithm was obtained by Gr\"{o}tschel et al.~\cite{grotschel1981ellipsoid} using the ellipsoid method. This was followed by a plethora of improvements, among which the more recent ones~\cite{dadush2018geometric, jiang2020improved} leveraged related techniques. On a different front, there has been significant work dedicated to obtaining strongly polynomial time algorithms for this problem~\cite{fleischer2003push, iwata2003faster, iwata2001combinatorial, orlin2009faster, schrijver2000combinatorial, lee2015faster, dadush2018geometric, jiang2020improved}.

For the more structured regime of \emph{decomposable submodular function minimization}, algorithms based on both discrete and continuous methods have been developed. Kolmogorov~\cite{kolmogorov2012minimizing} has shown that this problem reduces to computing maximum submodular flows, and gave an algorithm for this problem based on augmenting paths. This was followed by further algorithms based on discrete methods~\cite{arora2012generic, fix2013structured}. 
The continuous methods are based on convex optimization on the submodular base polytope, which is also used here. Notably,
Stobbe and Krause~\cite{stobbe2010efficient} tackled this problem using gradient descent, Nishihara et al.~\cite{nishihara2014convergence} used alternating projections to obtain an algorithm with linear convergence, Ene and Nguyen~\cite{ene2015random} achieved an improved algorithm with a linear convergence rate based on accelerated coordinate descent, while Ene et al. provided further improvements both via gradient descent and combinatorial techniques~\cite{ene2017decomposable}. 

\paragraph{Related Works on Parametric Min Cut}
The seminal work of Gallo et al.~\cite{GalloGT89} studied the generalization of the maximum flow problem where some edge-capacities, instead of being fixed, are allowed to be (possibly different) monotonic functions of a single parameter. They showed how to modify certain versions of the push-relabel algorithm for ordinary maximum flow to the parametric problem with the same asymptotic time complexity. In particular, using the Goldberg-Tarjan max-flow algorithm \cite{goldberg1988new} they gave an $O(nm \log (n^2/m))$ time bound for the parametric version. Their algorithm can compute the min cuts either when a set of parameter values are given \cite{GusfieldM92} or the capacity functions are all affine functions of the parameter $\lambda$.

Several other max-flow algorithms were also shown to fit into their framework (see e.g.,~\cite{GranotMQT12}) though all requiring $\Omega(mn)$ time in the worst case. Further generalizations of the parametric min-cut problems have also been considered~\cite{McCormick99, GranotMQT12}. When all parameterized capacities are equal to the parameter $\lambda$, Tarjan et al.~\cite{TarjanWZZM06} give a divide and conquer approach that can use any maximum flow algorithm as a black box and is a factor min$\{n, \log(nU)\}$ worse.

 \section{Background and Preliminaries}\label{sec:prelim}
\subsection{Notation}
We let $[n] \stackrel{\textnormal{def}}{=} \{1,\dots,n\}$. 
We write $\|\cdot\|_p$ for the $\ell_p$ norm, i.e. $\|x\|_p = \left(\sum_i |x_i|^p\right)^{1/p}$, with $\|x\|_\infty = \max_i |x_i|$.

\subsection{Submodular Set Functions and Convex Analysis}
Let $V$ be a finite ground set of size $n$, and we assume w.l.o.g.
that $V=\left\{ 1,\dots,n\right\} $. A set function $F:2^{V}\rightarrow\mathbb{R}$
is submodular if $F\left(A\right)+F\left(B\right)\geq F\left(A\cup B\right)+F\left(A\cap B\right)$
for any two sets $A,B\subseteq V$. We are concerned with minimizing
submodular set functions of the form $F=\sum_{i=1}^{r}F_{i}$, where
each $F_{i}$ is a submodular set function:
\[
\min_{A\subseteq V}F\left(A\right)=\min_{A\subseteq V}\sum_{i=1}^{r}F_{i}\left(A\right)\,.
\]

For the rest of the paper we will assume that $F_i$ are non-negative, integral, and that $\max_{S\subseteq V}F(S) \leq \fmax$. 
The non-negativity constraint holds without loss of generality, as we can simply shift each $F_i$ by a constant until it becomes non-negative. 
For rational functions that are represented on bounded bit precision, the integrality can be enforced simply by scaling them, at the expense of increasing $\fmax$. As we will see, some of our subroutines depend on the magnitude of $\fmax$, so we will generally assume that this is polynomially bounded.

As in previous works~\cite{nishihara2014convergence, ene2015random, ene2017decomposable, kumar2019fast}, in this paper we are concerned with the regime where each function $F_i$ in the decomposition acts on few elements of the ground set $V$. More precisely for each $i \in \{1,\dots,\m\}$ there is a \textit{small} set $V_i \subseteq V$ such that $F_i(A) = F_i(A \cap V_i)$ for all $A \subseteq V$. We assume w.l.o.g. that $F_i(\emptyset) = F_i(V_i)$, which we discuss in more detail in Section~\ref{sec:assumptions}. The running time of our algorithm depends on $\max_{1 \leq i \leq \m} \vert V_i \vert$. This assumption is important as, furthermore, the final running time of our algorithm depends on (i) the time $\opto_i$ to optimize functions of the form $F_i(S) + w(S)$ over $V_i$, where $w$ is a linear function and (ii) the time $\evo_i$ to evaluate $F_i$ for subsets of $V_i$. In the case where $\vert V_i \vert = O(1)$, this is also constant time.

Given an arbitrary vector $w\in\mathbb{R}^{n}$ and a subset $A\subseteq V$,
we use the notation $w\left(A\right)=\sum_{i\in A}w_{i}$.
\begin{defn}
Given a submodular set function $F:2^{V}\rightarrow\mathbb{R}$, such that $F(\emptyset) = 0$, its
submodular base polytope $B\left(F\right)$ is defined as follows:
\begin{align*}
B\left(F\right)=\{ w\in\mathbb{R}^{n}:w\left(A\right)&\leq F\left(A\right)\ \text{for all }A\subseteq V,\, %
w\left(V\right)=F\left(V\right)\} \ .
\end{align*}
\end{defn}

\begin{defn}
Given a submodular set function $F:2^{V}\rightarrow\mathbb{R}$, $F(\emptyset) = 0$, its
Lov\'{a}sz extension $f:\mathbb{R}^{n}\rightarrow\mathbb{R}$
is defined over $[0,1]^n$ as the convex closure of $F$. However, it will be more
convenient to consider its extension of $\mathbb{R}^n$, given
 by 
\begin{align*}
f\left(x\right)&=\int_{0}^{\infty} F\left(\left\{ i:x_{i}\geq t\right\} \right)dt
+ \int_{-\infty}^0 \left(F\left(\left\{ i:x_{i}\geq t\right\} \right) - F\left(V\right) \right) dt\,.
\end{align*}
\end{defn}

\begin{fact}
It is well known~\cite{bach2011learning} that the Lov\'{a}sz extension of a submodular
set function $F$ can be equivalently characterized in terms of its
submodular base polytope $B\left(F\right)$. More precisely, if $F(\emptyset) = 0$, then:
\[
f\left(x\right)=\max_{w\in B\left(F\right)}\left\langle w,x\right\rangle \,.
\]
\end{fact}

For parametric submodular function minimization we consider a family
of functions parameterized by $\alpha\in\mathbb{R}$:
\begin{equation}
F_{\alpha}\left(A\right)=F\left(A\right)+\sum_{i\in A}\psi_{i}'\left(\alpha\right)\,,\label{eq:parametric-def}
\end{equation}
 where $\psi_{j}:\mathbb{R}\rightarrow\mathbb{R}$ are strictly convex 
differentiable functions, satisfying $\lim_{\alpha \rightarrow -\infty} \psi_i'(\alpha)=-\infty$ and $\lim_{\alpha \rightarrow \infty} \psi_i'(\alpha) = \infty$, for all $i$. 
A common example is $\psi_j'(\alpha) = \alpha$, which imposes an $\ell_1$ penalty on the size of the set $A$.
It is shown in~\cite{chambolle2009total, bach2011learning} that minimizing $F_{\alpha}\left(A\right)$
for the entire range of scalars $\alpha$ amounts to minimizing a
regularized version of the Lov\'{a}sz extension. 
\begin{lem}
\label{lem:param-sfm}Let $F_{\alpha}$ be the family of parameterized
submodular set functions defined as in (\ref{eq:parametric-def}), where $\psi_i$ are strictly convex functions.
Let $f$ be the Lov\'{a}sz extension of $F$,
and consider the optimization problem
\begin{equation}
\min_{x \in \mathbb{R}^n}f\left(x\right)+\sum_{i\in V}\psi_{i}\left(x_{i}\right)\ .\label{eq:parametric-optimization}
\end{equation}
Let $A^{\alpha}=\arg\min_{A\subseteq V}F_{\alpha}\left(A\right)$, and let $x^*$ be the minimizer of (\ref{eq:parametric-optimization}).
Then
\begin{equation}
A^{\alpha}=\left\{ i:x^*_{i}\geq\alpha\right\} \label{eq:param_set_def}\ .
\end{equation}
\end{lem}

For completeness we reproduce the proof of Lemma \ref{lem:param-sfm}
in Section \ref{sec:appendix-bach-proof}. 

Via convex duality one can prove that minimizing (\ref{eq:parametric-optimization}) is equivalent to a dual optimization problem on the submodular base polytope $B(F)$:
\begin{equation}
\min_{w \in B(F)} \sum_{i\in V} \psi_i^*(-w_i)\,, \label{eq:parametric-optimization-dual}
\end{equation}
where $\psi^*_i$ is the \textit{Fenchel dual} of $\psi_i$.
\begin{defn}[Fenchel dual]
Let $g:\mathbb{R}^n\rightarrow\mathbb{R}\cup\{-\infty,+\infty\}$ be a convex function. Its Fenchel dual or convex conjugate    $g^*:\mathbb{R}^n\rightarrow\mathbb{R}\cup\{-\infty,+\infty\}$  is defined as
\begin{align*}
g^*(w) = \underset{x}{\sup}\, \langle w, x\rangle - g(x)\,.
\end{align*}
\end{defn}

We will refer to (\ref{eq:parametric-optimization})  as the primal problem and (\ref{eq:parametric-optimization-dual}) as the dual problem. The algorithm described in this paper will focus on optimizing (\ref{eq:parametric-optimization-dual}) while strongly leveraging the decomposable structure of $F$.
We assume that all functions $\psi_i$ have ``nice'' second derivatives, which will play an important role in the algorithm, since this will also ensure that the minimizers of (\ref{eq:parametric-optimization}) and (\ref{eq:parametric-optimization-dual}) are unique.

\begin{assumption}\label{defn:conditions}
The function $\psi_i$
is $\smo$-smooth and $\str$-strongly convex for all $i \in V$.
Equivalently for each $i$, its second derivative satisfies $0 < \str \leq \psi_i''(x) \leq \smo$, for all $x \in \mathbb{R}$. 
Furthermore, $\vert \psi_i'(0) \vert \leq n^{O(1)}$, for all $i \in V$.
\end{assumption}

This condition also helps us ensure that we can efficiently convert between the primal and dual spaces onto which the $\psi_i$ and its Fenchel dual $\psi^*_i$ act. Also, whenever it is convenient, we will use the notation $\psi(x) = \sum_{i\in V} \psi_i(x_i)$,  $\psi^*(y)=\sum_{i\in V} \psi^*_i(y_i)$.

\subsection{Overview of Approach}
\paragraph{Decomposable Submodular Minimization}
Here we provide an overview of our approach for minimizing decomposable submodular functions. Our approach for the parametric setting yields a strictly stronger result without sacrificing running time, so we will focus on this more general problem.

Our approach is based on minimizing a convex function on the submodular base polytope $B(F)$. As it has been seen in previous works~\cite{bach2011learning}, in order to solve the parametric problem (\ref{eq:parametric-def}), it suffices to solve the dual problem (\ref{eq:parametric-optimization-dual}), which is a convex optimization problem over $B(F)$.
For convenience let us denote by $h(w) = \sum_{i\in V} \psi^*_i(-w_i)$, so that our objective becomes computing $\min_{w \in B(F)} h(w)$.

We use an iterative method, which maintains a point $w \in B(F)$ and updates it in such a way that the objective value improves significantly in each step. To do so, we find a polytope $P$ such that 
\begin{equation} 
w+\frac{1}{\alpha} \cdot P \subseteq B(F) \subseteq w + P \label{eq:intro_sandwich}
\end{equation}
and such that we can efficiently minimize $\psi$ over $w+\frac{1}{\alpha}\cdot P$. If we can find the minimizer $w'$ over  $w+\frac{1}{\alpha} \cdot P$, then moving our iterate to $w'$ also guarantees that
\[
h(w') - h(w^*) \leq \left(1-\frac{1}{\alpha}\right) \left( h(w) - h(w^*)\right)\,,
\]
where $w^*$ is the minimizer of $h$ over $B(F)$.
This is true due to the convexity of $h$. Indeed, let $\widetilde{w} = w + t(w^* - w)$ where $t = \max\{t \leq 1 : w+t(w^*-w) \in w+\frac{1}{\alpha} P\}$; in other words $\widetilde{w}$ represents the furthest point on the segment connecting $w$ and $w^*$ such that $\widetilde{w}$ still lies inside the small polytope $w + \frac{1}{\alpha} P$.
Due to the sandwiching property of the polytopes (\ref{eq:intro_sandwich}), we have that $t \geq 1/{\alpha}$. Hence, using the convexity of $h$, we obtain that 
\begin{align*}
&h(\widetilde{w}) - h(w^*)
=
h(w+t(w^* - w)) - h(w^*)
\leq \left(1-t\right) (h(w) - h(w^*)) \leq (1-1/\alpha) (h(w) - h(w^*))\,.
\end{align*}
Since $w'$ minimizes $h$ over $w+\frac{1}{\alpha} \cdot B(F)$, we must have $h(w') \leq h(\widetilde{w})$, and we obtain the desired progress in function value. Thus iterating $\widetilde{O}(\alpha)$ times we obtain a high precision solution, which we then convert back to a combinatorial solution to the original problem using some careful error analysis.

More importantly, we need to address the question of finding a polytope $P$ satisfying (\ref{eq:intro_sandwich}). 

To do so,
for each $i$, we define the ``residual'' submodular functions $F'_i(A) = F_i(A) - w_i(A)$ for all $A \subseteq V$, 
where $w_i \in B(F_i)$  such that $\sum_{i=1}^{\m} w_i = w$. The existence of such a decomposition of $w\in B(\sum_{i=1}^r F_i)$ is well-known, and goes back to Edmonds~\cite{edmonds1971matroids}. Very importantly, we note that since $F_i$ were non-negative, $F'_i$ remain non-negative submodular set functions. 

It is known~\cite{devanur2013approximation} that non-negative submodular set functions can be approximated by graph cuts. Following the proof from~\cite{devanur2013approximation}, for each $i$ we construct a graph on $O(\vert V_i\vert)$ vertices whose cuts approximate the value of $F_i'$ within a factor of $O(\vert V_i \vert^2)$. Combining all these graphs into a single one, we  obtain a weighted directed graph on $O(\vert V \vert)$ vertices and $O(\vert V \vert + \sum_{i=1}^\m \vert V_i \vert^2)$ arcs such that its cut function $G$ approximates $F'$ within a factor of $O(\max_i \vert V_i \vert^2)$.

Crucially, we can show that if $G$ is the cut function which approximates $F'$, then we also have that
\[
\frac{1}{\max_i \vert V_i \vert^2} \cdot B(G) \subseteq B(F_i') \subseteq B(G)\,,
\]
 and therefore it suffices to implement a routine that minimizes $h$ over $w + \frac{1}{\max_i \vert V_i \vert^2} \cdot B(G)$ in order to obtain an algorithm that terminates in $\widetilde{O}(\max_i \vert V_i \vert^2)$ such iterations.

To implement this routine, we devise a new combinatorial algorithm for solving the parametric flow problem, with general parameterized capacities. 
By comparison to previous literature, our algorithm efficiently leverages a maximum flow oracle on a sequence of graphs obtained via contracting edges, and whose running time is up to polylogarithmic factors equal to that of computing a maximum flow in a  capacitated directed graph with $O(\vert V \vert)$ vertices and $O(\vert V \vert + \sum_{i=1}^r \vert V_i\vert^2 )$ arcs.

Following this, we convert the combinatorial solution to the parametric flow problem into a solution to its corresponding dual problem on the submodular base polytope, which returns the new iterate $w'$.

Throughout the algorithm we need to control the errors introduced by the fact that both the solution we receive for the parametric flow problem and the one we return as an approximate minimizer of $h$ over $B(F)$ are \textit{approximate}, but these are easily tolerable since our main routines return high precision solutions.

\section{Parametric Min $s,t$-Cut}\label{sec:parametric-reduction}
\newcommand{\dom}{\mathbb{D}}
\newcommand{\cutf}{\kappa}
\newcommand{\emb}{\tau}

In the general parametric min $s,t$-cut problem~\cite{GalloGT89}, the capacities
of the source's outgoing edges $(s,v)$ are (possibly different) nonnegative real nondecreasing functions of a parameter $\lambda\in\dom$, where $\dom\subseteq\mathbb{R}$ is some domain,
whereas the capacities of the sink's incoming edges $vt$ are nonincreasing functions of~$\lambda$. The goal is to compute the representation
of the \emph{cut function} $\cutf:\dom\to \mathbb{R}$ such
that $\cutf(\lambda)$ equals the capacity of the minimum $s,t$-cut in $G_{\lambda}$ obtained from $G$ by evaluating the parameterized capacity functions
at $\lambda$. It is known that $\cutf$ consists of $O(n)$ pieces, where
$\cutf$ equals the parameterized capacity of some fixed cut in $G$.

More formally, let $\lmin\in \dom$ be such that
the minimal min $s,t$-cuts of $G_{\lmin}$ and
$G_{\lambda'}$ are equal
for all $\lambda'\in \dom$, $\lambda'<\lmin$.
Similarly, let $\lmax\in\dom$ be such
that the minimal min $s,t$-cuts of $G_{\lmax}$ and~$G_{\lambda'}$ are equal for all $\lambda'\in\dom$ with $\lambda'>\lmax$.
We will consider $\lmin$ and $\lmax$ inputs
to our problem.
Then, there exist $O(n)$
\emph{breakpoints} $\Lambda=\{\lambda_1,\ldots,\lambda_k\}$, $\lmin=\lambda_0<\lambda_1<\ldots<\lambda_k$ and an embedding of vertices $\emb:V\to \Lambda\cup\{\lmin,\infty\}$ such that for all $i=0,\ldots,k-1$, $\lambda'\in [\lambda_i,\lambda_{i+1})\cap \dom$, $\cutf(\lambda')$ equals the capacity
of the cut $S(\lambda_i)=\{v\in V:\emb(v)\leq \lambda_i\}$ in $G_{\lambda'}$,
and also $S(\lambda_k)$ is a min $s,t$-cut
of $G_{\lmax}$.

Motivated by our submodular minimization application, our algorithm in the most general setting solves the 
$\eps$-approximate parametric min $s,t$-cut problem. %
\begin{defn}[$\eps$-approximate parametric min $s,t$-cut]\label{d:eps-pmstc}
Let $\Lambda$, $\emb$, and $S:\dom\to 2^V$ be
as defined above.
A pair $(\Lambda,\emb)$ is called an $\eps$-approximate parametric min
$s,t$-cut of $G$ if:
\begin{enumerate}%
    \item For $i=0,\ldots,k-1$, $S(\lambda_i)$ is a
    min $s,t$-cut of $G_{\lambda'}$ 
    for all $\lambda'\in [\lambda_i,\lambda_{i+1}-\eps)\cap \dom$.
    \item $S(\lambda_k)$ is a
    min $s,t$-cut of $G_{\lmax}$.
    \item For $i=0,\ldots,k-1$, $S(\lambda_i)\subsetneq S(\lambda_{i+1})$.
\end{enumerate}
\end{defn}

We prove that an $\eps$-approximate parametric
min $s,t$-cut yields breakpoints
within $\eps$ additive error wrt. to the breakpoints
of the exact parametric min $s,t$-cut.
Our algorithm solves the above problem
assuming only constant-time black-box access to the capacity functions.

\begin{restatable}{thm}{aprxpmc}
\label{t:aprx-pmc}
Let $R=\lmax-\lmin$ be an integral multiple of $\eps>0$. Let $\tmf(n',m')=\Omega(m'+n')$ be a convex function bounding the
time needed to compute maximum flow in a graph
with $n'$ vertices and $m'$ edges obtained from $G_\lambda$ by edge/vertex deletions and/or edge contractions (with merging parallel edges by summing their capacities) for any $\lambda=\lmin+\ell\eps$ and any integer $\ell\in [0,R/\eps]$.
Then, $\eps$-approximate parametric min $s,t$-cut in $G$ can
be computed in $O(\tmf(n,m\log{n})\cdot\log{\frac{R}{\eps}}\cdot \log{n})$ time.
\end{restatable}

The algorithm is recursive.
In order to ensure uniqueness of the minimum cuts considered, it always
computes cuts with \emph{minimal} $s$-side.

Roughly speaking, given initial guesses
$\lmin,\lmax$
the algorithm finds, using \linebreak $O(\log((\lmax-\lmin)/\eps))$ maximum
flow computations, the most balanced split $\lambda_1,\lambda_2$ of the domain
such that (1) the $s$-sides for all the min-cuts of $G_{\lambda'}$ for $\lambda'>\lambda_2$ have size at least $n/2$, (2)
the $t$-sides of all the min-cuts of $G_{\lambda'}$ for $\lambda'<\lambda_1$
have size at least $n/2$, (3) $\lambda_2-\lambda_1\leq \eps$.
We then recurse on the intervals $[\lmin,\lambda_1]$ and $[\lambda_2,\lmax]$
on \emph{minors} of~$G$ with at least $n/2$ vertices contracted.
Even though the contraction requires merging parallel edges
in order to have at most $m+n$ (as opposed to $2m$) edges in the
recursive calls, we are
able to guarantee that the capacity functions in the recursive
calls are all obtained by shifting the original capacity
functions by a real number, and thus can be evaluated in constant
time as well. Since the number of vertices
decreases by a factor of two in every recursive call, one can prove
that for each level of the recursion tree, the sum of numbers
of vertices in the calls at that level is $O(n)$, whereas
the sum of sizes of edge sets is $O(m+n\log{n})$. 

We show that the $\eps$-approximate algorithm can be turned
into an exact algorithm in two important special cases.
First of all, if the capacity functions are low-degree
(say, at most~$4$) polynomials with \emph{integer coefficients}
in $[-U,U]$, then one can compute parametric min $s,t$-cut
only $O(\polylog\{n,U\})$ factors slower than best known
max-flow algorithm for integer capacities~\cite{flow-sparse, flow-dense, goldberg1998beyond}.
Second, we can solve the \emph{discrete domain} case, i.e., when
$\dom$ has finite size $\ell$ with only $O(\polylog\{n,\ell\})$
multiplicative overhead wrt. the respective maximum flow algorithm.

Moreover, since our reduction runs maximum flow computations only on \emph{minors}
of the input graph~$G$, it also yields very efficient
parametric min $s,t$-cut algorithms for planar graphs.
In particular, since near-optimal strongly-polynomial
$s,t$-max flow algorithms for planar graphs are known~\cite{BorradaileK09, Erickson10}, we obtain near-optimal algorithms for
the integer polynomial capacity functions (as above) and
discrete domains. What is perhaps more surprising,
using our reduction we can even obtain a \emph{strongly polynomial}
exact parametric min $s,t$-cut algorithm for planar graphs
with linear capacity functions with real coefficients.
The algorithm runs in $\tilde{O}(n^{1.21875})$ time and constitutes
the only known subquadratic strongly polynomial parametric min-$s,t$-cut
algorithm.

The details of our parametric min $s,t$-cut algorithm and its applications are covered
in Appendix~\ref{sec:parm-st-cut}.
It should be noted that the idea of using cuts contraction is not new and
appeared previously in~\cite{TarjanWZZM06}. Compared to~\cite{TarjanWZZM06}, our reduction provably handles
more general parameterized capacity functions. As it does not
operate on any auxiliary networks that may not preserve structural properties of $G$, but merely on minors of $G$, it proves
much more robust in important special cases such as planar graphs.
Finally, we believe that our reduction is also more natural and operates
on the breakpoints of the cut function directly, whereas the reduction of~\cite{TarjanWZZM06}
operates on so-called balanced flows.

 \section{Parametric Decomposable Submodular Minimization via Base Polytope Approximations}
\label{sec:iterative}

\subsection{Algorithm Overview}
In Algorithm~\ref{alg:main} we give the description of our the main
routine.

\begin{algorithm} \caption{Parametric Decomposable Submodular Function Minimization}\label{alg:step_solve}
\begin{algorithmic}[1]
\STATE {\bfseries Input:} $\epsilon$: error tolerance \COMMENT{Returns $\epsilon$-optimal solution of $\underset{\yy\in B(F)}{\min}\, \psi^*(-\yy)$}
\STATE Set $\yy^{0,i} = 0$ for $i \in[\m]$ \COMMENT{Initialize a feasible solution}
\STATE Set $T = \C \log \frac{\psi^*(-\yy^0) + \psi(\zerov)}{\eps}$
\FOR {$t=1\dots T$}
	\STATE Set $G_i(V_i,E_i,c_i) = \textsc{GraphApprox}(\yy^{t-1,i})$, for $i\in[\m]$, and construct $G(V,E,c)$ by combining the graphs $G_i$ %
	\STATE Set $\phi(\xx) := \psi(\xx) + \langle \sum\limits_{i=1}^r\yy^{t-1,i}, \xx\rangle$
	\label{line:phi}
	\STATE Set $\tyy = \textsc{FindMinCuts}(G(V,E,c), \phi, \frac{1}{3 \smo})$ \COMMENT{Find parametric min $s,t$-cuts}
	\label{line:findmincuts}
	\STATE Round all the entries of $\tyy$ to the nearest integer
	\STATE Decompose $\ty = \sum_{i=1}^\m \ty^i$, with $\ty_i \in B(G_i)$ using Lemma~\ref{lem:cut2dual}.
	\STATE Set $\yy^{t,i} = \yy^{t-1,i} + \tyy^i$, for all $i\in[\m]$.
\ENDFOR
\STATE \RETURN $\sum\limits_{i=1}^r \yy^{T,i}$ %
\end{algorithmic}
\label{alg:main}
\end{algorithm}

\subsection{Removing Assumptions}
In Section~\ref{sec:prelim} we assumed that for all $i$, $F_i(\emptyset) = F_i(V_i) = 0$ and $F_i(S) \geq 0$ for all $S$.
These assumptions hold without loss of generality. In Section~\ref{sec:assumptions} we show how to preprocess the input such that these assumptions are valid.

\subsection{From Parametric Minimum Cut to Cut Base Polytope Optimization}

In this section, we will focus on the problem of minimizing a convex function over the base polytope of the $s,t$-cut function of a graph $G(V\cup\{s,t\},E,c)$.
We define the $s,t$-cut function $G:2^V\rightarrow\mathbb{R}$ as
$G(S) = c^+(S\cup \{s\})$ for $S\subseteq V$, and let $B(G)$ be the base polytope of $G$.
Now, the parametric min $s,t$-cut problem can be written as
\begin{align}
\underset{S\subseteq V}{\min}\, G(S) + \sum\limits_{u\in S} \phi_u'(\lambda) \,,
\label{eq:parametric_cut_problem}
\end{align}
where
the capacity of an edge $(u,t)$ at time $\lambda$ is
$c_{ut} + \phi_u'(\lambda)$ and $\phi$ is a function satisfying Assumption~\ref{defn:conditions}.
In particular, our goal in this section is, given solutions to (\ref{eq:parametric_cut_problem}) for all $\lambda$, to solve the following dual problem:
\begin{align}
\underset{\yy\in B(G)}{\min}\, \phi^*(-\yy)\,.
\end{align}

\begin{defn}[$W$-restricted function]
A submodular function $F:2^V\rightarrow \mathbb{R}_{\geq \zerov}$
is called $W$-restricted if
$F(S) = F(S\cap W)$ for all $S\subseteq V$, where $W\subseteq V$.
\end{defn}
As the cut functions $G(S)$ that we will be concerned with will be decomposable, i.e. $G(S) = \sum\limits_{i=1}^r G_i(S)$ for all $S\subseteq V$,
we introduce the following notion of a \emph{decomposition} of some $\yy \in B(G)$ into a sum of $\yy^i\in B(G_i)$.
\begin{defn}[$F$-decomposition]
Let $F:2^V\rightarrow\mathbb{R}_{\geq \zerov}$ with $|V|=n$ be a submodular function that is decomposable, i.e.
$F(S) = \sum\limits_{i=1}^r F_i(S)$
for all $S\subseteq V$,
where $F_i:2^{V}\rightarrow\mathbb{R}_{\geq \zerov}$ are submodular functions.
Then for any $\yy\in B(F)$ there exist vectors $\yy^1,\dots,\yy^r\in\mathbb{R}^n$, where $\yy^i\in B(F_i)$ for all $i\in[r]$ and
$\sum\limits_{i=1}^r \yy^i = \yy$.
We call the sequence of vectors $\yy^1,\dots,\yy^r$ an $(F_1,\dots,F_r)$-decomposition of $\yy$,
or just an $F$-decomposition of $\yy$ if the $F_i$'s are clear from context.
\end{defn}

What follows is the main lemma of this section, whose full proof appears in Appendix~\ref{proof:lem_cut2dual}.

\begin{lem}[From parametric min-cut to cut base polytope optimization]
Consider a graph $G(V,E,c\geq 0)$
and a function $\phi(\xx) = \sum\limits_{u\in V} \phi_u(x_u)$ that satisfies
Assumption~\ref{defn:conditions}.
Additionally, let $G(S) = c^+(S)$ for all $S\subseteq V$ be the cut function associated with the graph, and suppose
it is decomposable as
$G(S) = \sum\limits_{i=1}^r G_i(S)$ where
$G_i:2^{V}\rightarrow \mathbb{Z}_{\geq 0}$ %
are $V_i$-restricted cut functions
defined as $G_i(S) = c^{i+}(S)$ that correspond to graphs $G_i(V, E, c^i \geq 0)$,
and $c = \sum\limits_{i=1}^r c^i$.%

We define %
	an extended vertex set $V' = V\cup \{s,t\}$ with edge set $E'=E\cup \underset{u\in V}\bigcup (s,u) \cup \underset{u\in V}{\bigcup} (u,t)$, %
	the parametric capacity of an edge $(u,v)\in E'$ as
\begin{align*}
c_\lambda(u,v) = \begin{cases}
	\max\{0,\phi_u'(-\lambda)\} & \text{if $u\in V, v=t$}\\
	\max\{0,-\phi_u'(-\lambda)\} & \text{if $u=s, v\in V$}\\
c_{uv} & \text{otherwise}
\end{cases}
\end{align*}
and let $(\Lambda, \tau)$ be
a
$\frac{1}{3\smo}$-approximate parametric min $s,t$-cut of
	$G'(V',E',c_\lambda)$.

There exists an algorithm that, given $(\Lambda, \tau)$,
outputs
	$\tyy^* = \underset{\yy\in B(G)}{\mathrm{argmin}}\, \phi^*(-\yy)$ and
a $G$-decomposition
$\tyy^{*1},\dots,\tyy^{*r}$
of $\tyy^*$.
The running time of this algorithm is $O\left(n + \sum\limits_{i=1}^r |V_i|^2\right)$. \label{lem:cut2dual}
\end{lem}

\subsection{Dual Progress Analysis in One Step}

\begin{lem}[Dual progress in one step]
Let $F:2^V\rightarrow \mathbb{Z}_{\geq 0}$ be a submodular function that is separable, i.e.
$F(S) = \sum\limits_{i=1}^r F_i(S)$ for all $S\subseteq V$, where $F_i:2^{V}\rightarrow \mathbb{Z}_{\geq 0}$ are $V_i$-restricted submodular functions
	with $F_i(\emptyset)=0$.
Additionally, let $\psi:\mathbb{R}^n\rightarrow\mathbb{R}$ be a function that satisfies Assumption~\ref{defn:conditions}, where $|V|=n$.%

Given a feasible dual solution $\yy\in B(F)$
and an $F$-decomposition $\yy^1,\dots,\yy^r\in \mathbb{Z}^n$ of $\yy$,
there is an algorithm that outputs a vector
$\yy'\in \mathbb{Z}^n$, along with an $(F_1,\dots,F_r)$-decomposition
$\yy'^1,\dots,\yy'^r\in \mathbb{Z}^n$ of $\yy'$,
such that
\begin{align*}
\psi^*(-\yy') - \psi^*(-\yy^*)
&
\leq \left(1 - \frac{1}{\C}\right)\left(\psi^*(-\yy) - \psi^*(-\yy^*)\right) %
\end{align*}
where $\alpha = \underset{i\in[r]}{\max}\,\{|V_i|^2/4\}$ and
$\yy^* = \underset{\yy^*\in B(F)}{\mathrm{argmin}}\, \psi^*(-\yy^*)$ is the dual optimum.
The running time of the algorithm is
\begin{align*}&\widetilde{O}\left(\sum\limits_{i=1}^r |V_i|^2 \opto_i
+ \tmf\left(n,n + \sum\limits_{i=1}^r |V_i|^2\right)\right)\,.\end{align*}
\label{lem:progress_lemma}
\end{lem}
The full proof of Lemma~\ref{lem:progress_lemma} appears in Appendix~\ref{proof:lem_progress_lemma}.

\subsection{Main Theorem}

\ifx 0
\begin{thm}[Parametric Submodular Function Minimization]
Let $F:2^V\rightarrow\mathbb{R}_{\geq 0}$ be a submodular function that is separable, i.e.
$F(S) = \sum\limits_{i=1}^r F_i(S\cap V_i)$ for all $S\subseteq V$, where $F_i:2^{V_i}\rightarrow\mathbb{R}_{\geq 0}$ and $V_i\subseteq V$.
Additionally, let $\psi:\mathbb{R}^n\rightarrow\mathbb{R}$ be a function that satisfies Assumption~\ref{defn:conditions}, where $|V|=n$.
There is an algorithm,
running in time $\tO{}$,
that optimizes the following objective for all $\lambda\in\mathbb{R}$:
\begin{align*}
\underset{S\subseteq V}{\min}\, F(S) + \sum\limits_{u\in S}\psi'(\lambda)\,.
\end{align*}
It achieves this by producing a vector $\xx$ such that for any $\lambda\in\mathbb{R}$, the set
$S^\lambda = \{u\ :\ x_u \geq \lambda\}$ satisfies
\begin{align*}
F(S^\lambda) + \sum\limits_{u\in S^\lambda} \psi'(\lambda) \leq
\underset{S\subseteq V}\min\, F(S) + \sum\limits_{u\in S} \psi'(\lambda) + \eps\,.
\end{align*}
\end{thm}
\fi
\begin{proof}[Proof of Theorem~\ref{thm:main}]
We repeatedly apply Lemma~\ref{lem:progress_lemma},
	and let $\yy^0,\dots,\yy^T$ be the iterates
after $T = \C \log \frac{\psi^*(-\yy^0) - \psi^*(-\yy^*)}{\zeta}$ iterations. We have
\begin{align*}
\psi^*(-\yy^T) - \psi^*(-\yy^*)
&\leq  \left(1- \frac{1}{\C}\right)(\psi^*(-\yy^{T-1}) - \psi^*(-\yy^*)) \,.
\end{align*}
Applying induction over $T$ steps we obtain that
\begin{align*}
\psi^*(-\yy^T) - \psi^*(-\yy^*)
&\leq  \left(1- \frac{1}{\C}\right)^T(\psi^*(-\yy^0) - \psi^*(-\yy^*))
\leq  e^{-T/\C}(\psi^*(-\yy^0) - \psi^*(-\yy^*)) %
\leq \zeta\,.
\end{align*}
We have obtained a high precision solution to the objective (\ref{eq:parametric-optimization-dual}).
Finally,
setting $\frac{1}{\zeta} = \mathrm{poly}(\frac{\smo}{\str} n \fmax / \eps) = {n^{O(1)}}/{\eps}$
and applying Corollary~\ref{cor:dual2disc} to convert
from this solution to the actual sets, we obtain the desired solution.

As Assumption~\ref{defn:conditions} implies $\psi^*(-\yy^0) - \psi^*(-\ww^*) = n^{O(1)}$,
the total running time is
\begin{align*}
\widetilde{O}\bigg( \max_i \vert V_i \vert^2 \bigg( & \sum_{i=1}^{\m} \vert V_i \vert^2 \opto_i
+ \tmf\bigg(n, n+\sum_{i=1}^\m \vert V_i \vert^2\bigg) \bigg) \ln \frac{1}{\epsilon}\bigg)\,.
\end{align*}
\end{proof}

\section*{Acknowledgements}
KA was supported in part by the NSF grants CCF-1553428 and CNS-1815221.
AK, AM and PS were supported in part by ERC CoG project TUgbOAT no 772346.
We are grateful to Michael B. Cohen for discussions on the potential of second order methods for submodular minimization, which motivated parts of this project. We thank Alina Ene for pointing us to the relevant material from~\cite{bach2011learning}.

\newpage

\appendix

\section{Approximating Submodular Set Functions with Graph Cuts}
It is shown in~\cite{devanur2013approximation} that submodular set functions defined on a ground set of $n$ elements can be $O\left(n^{2}\right)$
approximated by directed graph cuts. We state this fact as a lemma,
and we include the proof for completeness in Section \ref{sec:proof-of-lemma-graph-approx} below.

\begin{defn} Given a submodular set function $F:2^V \rightarrow \mathbb{R}$, such that $F(\emptyset) = F(V) = 0$, and a weighted directed graph $G = (V, E, c)$, we say that the cut function of $G$ $\alpha$-approximates $F$ if 
\[
\frac{1}{\alpha} \cut( A  )  \leq F(A) \leq \cut( A  ) \,, \textnormal{for all } A \subseteq V\,.
\]
\end{defn}

\begin{algorithm} \caption{Approximate non-negative submodular function
$F=F^0-w^0$ by graph cuts, where $F^0 : 2^V \rightarrow \mathbb{R}_{\geq 0}$ 
is the initial submodular function and the shift vector $w^0 : V \rightarrow \mathbb{R}$ is given as input}\label{alg:graph_approx}
\begin{algorithmic}[1]
\STATE {\bfseries function }{\textsc{GraphApprox}$(w^0 : V \rightarrow \mathbb{R})$}
\Indent
\STATE Call \textsc{GraphApproxShifted}$(F^0-w^0)$
\EndIndent

\vspace{2mm}
\STATE {\bfseries function }{\textsc{GraphApproxShifted}$(F : 2^V \rightarrow \mathbb{R}_{\geq 0})$}
\Indent
\STATE Let $E = \{(u,v) \in V \times V : u\neq v\}$.
\FOR {$u, v \in V : u \neq v$}
		\STATE Compute $w_{uv} = \min_{\substack{A \subseteq V:\\ u\in A, v\not\in A}} F(A)$.
		\STATE $c_{uv} = w_{uv}$\,.
\ENDFOR
\RETURN $G = (V, E,c)$
\EndIndent
\end{algorithmic}
\end{algorithm}

\begin{lem}
\label{lem:graph-approx}Let $V=\left\{ 1,\dots,n\right\} $, and
let $F:V\rightarrow\mathbb{R}$ be a non-negative submodular set function, satisfying $F(\emptyset) = F(V) = 0$. Using
$O\left(n^{2}\right)$ calls to a minimization oracle which can compute
for all pairs $u,v\in V$
\[
\min_{\substack{A\subseteq V\\u\in A,v\not\in A} }F\left(A\right)
\]
one can compute a weighted directed graph $G\left(V,E,c\right)$
such that its cut function
\[
\cut\left(A\right):=\sum_{\substack{\left(u,v\right)\in E:\\
u\in A ,v\not\in A 
}
}c_{uv}
\]
$(n^2/4)$-approximates $G$.
In other words, for any $A\subseteq V$ the size of the graph cut 
satisfies:
\[
\frac{1}{n^2/{4}}\cdot \cut\left(A\right)\leq F\left(A\right)\leq\cut\left(A\right)\ .
\]
Furthermore, %
if $F$ takes only values that are discrete multiples of $\Delta$, i.e. $F(A) \in \Delta \cdot \mathbb{Z}_{\geq 0}$ for all $A$, then all elements of $c$ are discrete multiples of  $\Delta$.
\end{lem}

As a consequence, we obtain a good approximation by graph cuts for decomposable submodular functions where each component in the decomposition acts on few elements, i.e., when $F_i(A)=F_i(A\cap V_i)$
for some $V_i\subseteq V$.
\begin{lem}\label{lem:approx-sum-submod}
Let $V=\left\{ 1,\dots,n\right\} $, and
let $F_i:V_i\rightarrow\mathbb{R}$, $V_i \subseteq V$ be non-negative submodular set functions, with $F_i(\emptyset) = F_i(V_i)$, for $i=1,\dots,r$. In the time required to compute 
for all pairs $u\neq v\in V$ and for all $1\leq i \leq r$
\[
\min_{\substack{A\subseteq V_i\\u\in A,v\not\in A} }F_i\left(A\right)
\]
one can compute a weighted directed graph $G\left(V,E,c\right)$
such that its cut function $(M^2/4)$-approximates $\sum_{i=1}^\m F_i$, where $M = \max_{i=1,\m} \vert V_i \vert$.
\end{lem}
\begin{proof}
For each $i$ compute the corresponding graph as in Lemma~\ref{lem:graph-approx}. Then take the union of edges over the same vertex set.
\end{proof}  

We showed that the function $F(A)$ is well approximated by the cut function $c^+(A)$ for the graph we constructed. Note that $c^+$ is only defined on internal vertices of the graph, excluding $s$ and $t$. However this does not affect its submodularity. Therefore the submodular base polytopes for the two function approximate each other well.
\begin{lem}\label{lem:submod-base-polytope-approx} Let $F, G$ be two submodular functions defined over the same vertex set $V$ such that $F(\emptyset) = G(\emptyset) = 0$, $F(V) = G(V) = 0$, and for any $A \subseteq V$, $\frac{1}{\alpha} G(A) \leq F(A) \leq G(A)$. Then their submodular base polytopes satisfy:
$$\frac{1}{\alpha} B(G) \subseteq B(F) \subseteq B(G)\,.$$
\end{lem}
\begin{proof}
Let any $w \in B(F)$. Then $w(V) = F(V) = G(V)$. Furthermore for any set $A\subseteq V$, we have 
$w(A) \leq F(A) \leq G(A)$. Similarly for any $w\in G(A)$, we have $w(A) \leq G(A) \leq \alpha F(A)$, so $B(G) \subseteq \alpha B(F)$, which yields the claim.
\end{proof}

At this point we can prove that the submodular base polytope of the cut function created in Lemma~\ref{lem:approx-sum-submod} approximates the submodular base polytope of the decomposable function $\sum_{i=1}^r F_i$.

\begin{lem}
Let $V=\left\{ 1,\dots,n\right\} $, let $F_i:V_i\rightarrow\mathbb{R}$, $V_i \subseteq V$ be non-negative submodular set functions, with $F_i(\emptyset) = F_i(V_i)$, for $i=1,\dots,r$, and let $F=\sum_{i=1}^r F_i$. In the time required to solve $\min_{A\subseteq V_i:u\in A,v\not\in A} F_i(A)$ for all  $u,v\in V_i$ and all $i$, we can compute a weighted directed graph $G=(V, E,c)$ such that the submodular base polytope of the cut function $\cut(A)$ satisfies
$\frac{1}{M^2 / 4} B(\cut) \subseteq B(F) \subseteq B(\cut)$, where $M = \max_{i=1,\m} \vert V_i \vert$.
\label{lem:approx-cut}
\end{lem}
\begin{proof}
The proof follows directly from applying Lemma~\ref{lem:approx-sum-submod}, followed by Lemma~\ref{lem:submod-base-polytope-approx}.
\end{proof}

\subsection{Proof of Lemma \ref{lem:graph-approx}}\label{sec:proof-of-lemma-graph-approx}
\begin{proof}
To simplify notation let us denote by
\[
w_{uv} = \min_{\substack{A\subseteq V\\u\in A,v\not\in A} }F\left(A\right)\,,
\]
and let $T_{uv}$ be the set achieving this minimum. 

Consider the graph defined as follows. For every $u,v\in V$, create
an arc $\left(u,v\right)$ with weight $c_{uv}=w_{uv}$. 
By construction all capacities 
are discrete multiples 
of $\Delta$. 

Now we can prove the lower bound on $F$. We have that
\begin{align*}
\cut\left(A\right) 
=\sum_{u\in A,v\not\in A} c_{uv}  %
\leq \sum_{u\in A,v\not\in A} c_{uv}  %
=\sum_{u\in A,v\not\in A} F(T_{uv}) %
\leq\sum_{u\in A,v\notin A}F\left(A\right) %
\leq \left(n^{2}/4\right) F(A)\,.
\end{align*}
We used the fact that $F\left(A\right)$ upper bounds $c_{uv}$ for all $u \in A, v\not\in A$.
Now we prove the upper bound. For any nonempty set $A \subset V $ we can write 
$A = \bigcup_{u \in A} \left(  \bigcap_{v \in V\setminus A} T_{uv} \right)$. By twice applying Lemma~\ref{lem:submod_union_inter}, we obtain that
\[
F\left(A\right) \leq \sum_{u \in A} \sum_{v \not\in A} F\left(T_{uv} \right) = c^+\left(A\right)\,.
\]
Additionally, we have by construction that 
\[
\cut(\emptyset) = \cut(V) = 0\,.
\]
\end{proof}

\begin{lem}\label{lem:submod_union_inter}
Let $F$ be a non-negative submodular set function $F:2^V \rightarrow \mathbb{R}$, an let $A_1, \dots, A_t$ be subsets of $V$. Then
\[
F\left(\bigcup_{i=1}^t A_i\right) \leq \sum_{i=1}^t F(A_i)\,
\]
and
\[
F\left(\bigcap_{i=1}^t A_i\right) \leq \sum_{i=1}^t F(A_i)\,.
\]
\end{lem}
\begin{proof}
We prove by induction on $t$. If $t = 1$, both inequalities are equalities. Otherwise, suppose they hold for $t-1$. Let $S = \bigcup_{i=1}^{t-1} A_i$. By submodularity, $F(S \cup A_t) \leq F(S) + F(A_t) - F(S \cap A_t)$.
Since $F$ is non-negative, so is $F(S \cap A_t)$, and therefore $F(S \cup A_t) \leq F(S) + F(A_t)$. Applying the induction hypothesis this concludes the first part of the proof.

Similarly, let $S = \bigcap_{i=1}^{t-1} A_i$. By submodularity, $F(S \cap A_t) \leq F(S) + F(A_t) - F(S \cup A_t)$. Since $F$ is non-negative, so is $F(S \cup A_t)$, and therefore $F(S \cap A_t) \leq F(S) + F(A_t)$. Again, applying the induction hypothesis this concludes the second part of the proof.
\end{proof}

\section{Parametric Submodular Minimization via Optimization on the Base Polytope}\label{sec:appendix-bach-proof}
\label{sec:conversion-lemmas}
In this section, for completeness, we provide a proof of Lemma~\ref{lem:param-sfm}, which is based on~\cite{bach2011learning} (see Chapter 8).
In addition, we provide error analysis for reductions between approximate solutions to the combinatorial parametric submodular minimization problem, its continuous version involving the Lov\'{a}sz extension, and the dual formulation on the base polytope.

\begin{proof}[Proof of Lemma~\ref{lem:param-sfm}]
Given any point $x$, let $\beta \leq \min\{0, \min_i x_i\}$. Applying the definition of the Lov\'{a}sz extension, and the fundamental theorem of calculus, we can write:
\begin{align*}
f(x) + \sum_{i\in V}\psi_i(x_i) &= \int_0^\infty F(\{i : x_i \geq t\}) dt 
+ \int_{\beta}^0 \left(  F(\{i : x_i \geq t\}) - F(V)   \right) dt 
\\
&+ \sum_{i\in V}\psi_i(\beta) + \int_\beta^\infty \sum_{i: x_i \geq t} \psi_i'(t) dt
\\
&= \int_\beta^\infty \left(  F(\{i : x_i \geq t\}) + \sum_{i : x_i \geq t} \psi_i'(t) \right) dt + \sum_{i\in V}\psi_i(\beta) - \beta F(V)\,.
\end{align*}
Note that we crucially used the fact that the parametric term $\sum_i \psi_i'(t)$ is separable.

Next we show that if the optimal sets $A^\alpha$ were different from those defined in (\ref{eq:param_set_def}), then we could obtain a different iterate $x'$  such that $f(x')+\sum_{i\in V}\psi_i(x')\leq f(x^*) + \sum_{i\in V}\psi_i(x^*)$. However, since $\psi$ is strictly convex, the minimizer of $f(x)+\sum_{i\in V}\psi_i(x)$ is unique. This gives a contradiction leading us to the desired conclusion.

Indeed, let $x'_i = \sup_{i \in A^\alpha} \alpha$. By the strict convexity property of $\psi_i$, we have that for any
$\alpha > \beta$, $A^\alpha \subseteq A^\beta$, which we reprove for completeness in Lemma~\ref{lem:inclusion_chain}.

Using this fact, we know that if $A^\alpha$ are the optimizers of $F(A) + \sum_{i\in V}\psi_i(\alpha)$, then we can write:
\begin{equation*}
\int_\beta^\infty \left( F(A^t) + \sum_{i\in A^t} \psi_i'(t) \right) dt 
= \int_\beta^\infty \left( F( \{ i : x'_i \geq t \} ) + \sum_{i: x'_i \geq t} \psi_i'(t) \right) dt \,.
\end{equation*}
Since by the optimality of $A^t$ we have that 
\[
f(A^t)+\sum_{i\in A_t}\psi_i(t) \leq F(\{i : x_i \geq t\}) + \sum_{i : x_i \geq t} \psi_i'(t)\,,
\]
it means that letting $\beta = \min\{0, \min_i x'_i, \min_i x^*_i\}$,
\[
\int_\beta^\infty \left( F( \{ i : x'_i \geq t \} ) + \sum_{i: x'_i \geq t} \psi_i'(t) \right) dt 
\leq
\int_\beta^\infty \left( F( \{ i : x^*_i \geq t \} ) + \sum_{i: x^*_i \geq t} \psi_i'(t) \right) dt 
\]
and therefore
\[
f(x')+\sum_{i\in V}\psi_i(x')\leq f(x^*) + \sum_{i\in V}\psi_i(x^*)\,,
\]
which concludes the proof.
\end{proof}

\begin{lem} \label{lem:inclusion_chain}
Let $F:2^V \rightarrow \mathbb{R}$ be a submodular set function, and let $\psi_i : \mathbb{R} \rightarrow \mathbb{R}$ be a family of strictly convex functions, for $i\in V$. Let $F_\alpha(A) = F(A) + \sum_{i\in A} \psi_i'(\alpha)$, and $A^\alpha = \arg\min_{A\subseteq V} F_\alpha(A)$. If $\alpha > \beta$, then $A^\alpha \subseteq A^\beta$.
\end{lem}
\begin{proof}
By optimality we have that
\[
F(A^\alpha) + \sum_{i \in A^\alpha} \psi_i'(\alpha) \leq F(A^\alpha \cap A^\beta) + \sum_{i \in A^\alpha \cap A^\beta} \psi_i'(\alpha)
\]
and
\[
F(A^\beta) + \sum_{i \in A^\beta} \psi_i'(\beta) \leq F(A^\alpha \cup A^\beta) + \sum_{i \in A^\alpha \cup A^\beta} \psi_i'(\beta)\,.
\]
Summing up we obtain that
\begin{align*}
&\sum_{i \in A^\alpha} \psi_i'(\alpha) - \sum_{i \in A^\alpha \cap A^\beta} \psi_i'(\alpha) +  \sum_{i \in A^\beta} \psi_i'(\beta) - \sum_{i \in A^\alpha \cup A^\beta} \psi_i'(\beta) \\
&\leq F(A^\alpha \cap A^\beta) +F(A^\alpha \cup A^\beta)-F(A^\alpha) -  F(A^\beta)  \\
&\leq 0\,,
\end{align*}
where we used submodularity in the last step. Therefore
\[
\sum_{i \in A^\alpha \setminus A^\beta} (\psi_i'(\alpha) -\psi_i'(\beta)) \leq 0\,.
\]
Hence we conclude that $A^\alpha \setminus A^\beta = \emptyset$, since by strict convexity we have that for all $i$, $\psi_i'(\alpha) > \psi_i'(\beta)$, which would make the term above strictly positive had there been any elements in the set difference.

\end{proof}

Next we perform a careful error analysis to bound the total error we incur in the case  
where the iterate we consider is not an exact minimizer of (\ref{eq:parametric-optimization}), but has some small error in norm.

\begin{lem}\label{lem:primal-error-norm-to-set}
Under the conditions from Lemma~\ref{lem:param-sfm}, let $\widetilde{x} \in \mathbb{R}$ be a point satisfying
$\|\widetilde{x}-x^*\| \leq \epsilon$, where $x^*$ is the minimizer of (\ref{eq:parametric-optimization}). Let the sets 
\[
\widetilde{A}^\alpha = \{i : \widetilde{x}_i \geq \alpha\}\,.
\]
If $\psi_i$ is $\sigma$-strongly convex, for all $i$, and $\max_{A\subseteq V} F(A) - \min_{A'\subseteq V} F(A') \leq M$, then:
\[
F(\widetilde{A}^\alpha) + \sum_{i \in \widetilde{A}^\alpha} \psi_i'(\alpha) \leq F(A^\alpha) + \sum_{i\in A^\alpha} \psi_i'(\alpha) + Mn^{3/2}\epsilon + \beta\epsilon^2/2\,.
\]
\end{lem}
\begin{proof}
First, using the smoothness of $\psi_i$ we prove that
\[
f(\widetilde{x}) + \sum_i \psi_i(\widetilde{x}_i) \leq f(x^*) + \sum_i \psi_i(x^*_i)  + Mn^{3/2}\epsilon + \beta\epsilon^2/2\,.
\]
To prove this we first note that $f$ is Lipschitz, since we can use the fact that entries of the gradient of the Lov\'{a}sz extension consist of differences between $F$ evaluated at different subsets of $V$. Hence for any $x$, $\vert \nabla_i f(x) \vert \leq \max_{A \subseteq V} F(A) - \min_{A' \subseteq V} F(A') \leq M$, and thus $\| \nabla f(x) \| \leq M\sqrt{n}$. Therefore 
\[
f(\widetilde{x}) - f(x^*) \leq M\sqrt{n} \|\widetilde{x} - x^*\| \leq M\sqrt{n} \epsilon\,.
\]
Secondly, we use the smoothness of $\psi_i$, to obtain that 
\begin{align*}
\psi_i(\widetilde{x}_i) \leq \psi_i(x^*_i) + \psi_i'(x^*_i) (\widetilde{x}_i - x^*_i) + \frac{\beta}{2} (\widetilde{x}_i - x^*_i)^2 
\end{align*}
Using Lemma~\ref{lem:primal-dual-conversion-exact} we see that $\psi_i'(x^*_i) = -w^*_i$, where $w^*$ is the optimizer of a certain function over the base polytope $B(F)$. By the definition of $B(F)$ we have $w^*_i \leq F(\{i\}) \leq M$ and $-w^*_i + \sum_{j\neq i} w^*_j = F(V)$, so $-w^*_i \geq F(V) - \sum_{j\neq i} F(\{j\}) \geq -M(n-1)$.
Thus, by applying Cauchy-Schwharz, we have 
\[
\sum_{i \in V} \psi_i'(x^*_i) (\widetilde{x}_i - x^*_i) \leq \max_i \vert\psi_i'(x^*_i)\vert \sqrt{n} \cdot \|\widetilde{x} - x^*\|
\leq M(n-1) n^{1/2} \epsilon\,,
\]
and thus 
\[
\sum_{i \in V} \psi_i'(\widetilde{x}_i) - \psi_i'(x^*_i) \leq  M(n-1) n^{1/2} \epsilon + \beta \epsilon^2 / 2\,.
\]
Combining with the bound on $f(\widetilde{x})$, we obtain our claimed error in function value.

Now we can finalize the argument. Following the proof of Lemma~\ref{lem:param-sfm} we write $f(\widetilde{x})+\sum_{i\in V}\psi_i'(\widetilde{x}_i)$ as an integral, and similarly for $x^*$, to conclude that for $\beta = \min\{0, \min_i \widetilde{x}_i, \min_i x^*_i\}$,
\[
\int_\beta^\infty \left( F(\widetilde{A}^t) + \sum_{i\in \widetilde{A}^t} \psi_i'(t) \right) dt \leq \int_\beta^\infty \left( F(A^t) + \sum_{i\in A^t} \psi_i'(t) \right) dt + Mn^{3/2}\epsilon + \beta\epsilon^2/2\,.
\]
Since by definition $A^t$ minimizes $\sum_{i\in A^t} \psi_i'(t)$, we conclude that for all $t$,
\[
F(\widetilde{A}^t) + \sum_{i\in \widetilde{A}^t} \psi_i'(t) \leq  F(A^t) + \sum_{i\in A^t} \psi_i'(t) + Mn^{3/2}\epsilon + \beta\epsilon^2/2\,.
\]
\end{proof}

We can also show that if we obtain an approximate minimizer of the dual problem (\ref{eq:parametric-optimization-dual}) over $B(F)$, we can use it to recover an approximate minimizer of the primal problem (\ref{eq:parametric-optimization}).

\begin{lem}\label{lem:dual-to-primal-error-norm}
Let $w^*$ be the minimizer of the dual problem (\ref{eq:parametric-optimization-dual}), and 
and let $x^*$ be the minimizer of the primal problem (\ref{eq:parametric-optimization}).
If $w \in B(F)$ such that 
\[
\sum_{i\in V}\psi_i^*(-w_i) \leq \sum_{i\in V}\psi_i^*(-w^*_i) + \epsilon\,,
\]
then the point $x \in \mathbb{R}^n$ where $x_i = (\psi_i^*)'(-w_i)$ satisfies
\[
\|x-x^*\| \leq \sqrt{\frac{2 \smo \epsilon}{\str^2}}\,.
\]
\end{lem}
\begin{proof}
By Lemma~\ref{lem:primal-dual-conversion-exact} we know that $x^*$ and $w^*$ are related via $x^*_i = (\psi^*_i)'(-w^*_i)$. Therefore we can write
\[
\vert x_i - x^*_i \vert = \vert (\psi^*_i)'(-w_i) - (\psi^*_i)'(-w^*_i) \vert \leq  \frac{1}{\str}\vert w_i - w^*_i \vert\,,
\]
where in the last inequality we used the fact that $\psi_i$ is $\str$-strongly convex, and hence $\psi^*_i$ is $1/\str$-smooth~\cite{shalev2006convex, kakade2012regularization}.
Next we show that $\vert w_i - w_i^*\vert$ is bounded by a function of $\epsilon$.

Since by assumption $\psi_i$ is $\smo$-smooth, its dual $\psi^*_i$ is $1/\smo$-strongly convex. Therefore we have that, for all $i$:
\[
\psi^*_i(-w_i) \geq \psi^*_i(-w^*_i) + (\psi^*_i)'(-w^*_i) \cdot (-w_i - (-w^*_i)) + \frac{\str}{2} (w^*_i - w_i)^2\,.
\]
Furthermore, since $w^*$ is an optimizer over $B(F)$, we know by first-order optimality that for any $w \in B(F)$:
\[
\sum_{i\in V} (\psi^*_i)'(-w^*_i) \cdot (-w_i - (-w^*_i)) \geq 0\,,
\]
i.e. slightly moving the point from $-w^*$ towards $-w$ can only increase function value. Thus we obtain that
\[
\sum_{i \in V} \psi^*_i(-w_i) \geq \sum_{i \in V} \psi^*_i(-w^*_i) + \frac{1}{2\smo} \sum_{i\in V} (w^*_i - w_i)^2\,.
\]
Combining with the hypothesis, this implies that 
\[
\frac{1}{2\smo} \sum_{i\in V} (w^*_i - w_i)^2 \leq \epsilon\,,
\]
and therefore
\[
\|x-x^*\|^2 \leq \frac{1}{\str^2} \sum_{i \in V} (w_i - w^*_i)^2 \leq \frac{2 \smo \epsilon}{\str^2} \,,
\]
 which implies the claimed result.
\end{proof}

As a corollary of the previous lemmas, we see that an approximate solution to the dual problem (\ref{eq:parametric-optimization-dual}) yields an approximate solution to the original parametric problem (\ref{eq:parametric-def}).
\begin{corollary}
Let $F : 2^V  \rightarrow \mathbb{R}$ be a non-negative submodular set function, and let the the family of parametric problems defined in (\ref{eq:parametric-def}). Let $w \in B(F)$ such that  
\[
\sum_{i\in V}\psi_i^*(-w_i) \leq \sum_{i\in V}\psi_i^*(-w^*_i) + \epsilon\,,
\]
where $w^*$ is the true minimizer of the dual problem (\ref{eq:parametric-optimization-dual}). Then for any $\alpha$, the set
\[
\widetilde{A}^\alpha = \{i : \psi^*_i(-w_i) \geq \alpha\}
\]
satisfies
\[
F_\alpha( \widetilde{A}^\alpha ) \leq F_\alpha ( A^\alpha ) + \sqrt{\epsilon} \cdot Mn^{3/2} \sqrt{2\smo/\str^2} + \epsilon\cdot (\smo/\str)^2 \,,
\]
where $A^\alpha = \arg\min_{A \subseteq V} F_\alpha(A)$.
\label{cor:dual2disc}
\end{corollary}
\begin{proof}
From Lemma~\ref{lem:dual-to-primal-error-norm} we know that the hypothesis implies that the point $x$ where $x_i = (\psi^*_i)(-w_i)$ satisfies $\|x-x^*\| \leq \sqrt{2\smo\epsilon/\str^2}$. Applying Lemma~\ref{lem:primal-error-norm-to-set} we thus obtain that the sets constructed satisfy
\[
F_\alpha(\widetilde{A}^\alpha) \leq F_\alpha(A^\alpha) + Mn^{3/2}\sqrt{2\smo\epsilon/\str^2} + \smo / 2 \cdot (2\smo \epsilon/\str^2)\,,
\]
which yields our claim.
\end{proof}

The following helper lemma shows that we can efficiently convert between (exact) solutions to the primal and dual problems (\ref{eq:parametric-optimization}) and (\ref{eq:parametric-optimization-dual}). Using standard techniques we can prove that these also enable us to convert between suboptimal solutions, while satisfying certain error bounds.

\begin{lem}\label{lem:primal-dual-conversion-exact}
Let $x^*$ be the (unique) minimizer of (\ref{eq:parametric-optimization}), and let $w^*$ be the minimizer of (\ref{eq:parametric-optimization-dual}). Then
$w^*_i = -\psi_i'(x_i)$ and $(\psi^*_i)'(-w_i) = x_i$, for all $i \in V$.
\end{lem}
\begin{proof}
We the dual characterization of $f$ and Sion's theorem, to write
\begin{align*}
\min_{x \in \mathbb{R}^n} f(x) + \sum_{i \in V} \psi_i(x_i) = \min_{x \in \mathbb{R}^n} \max_{w \in B(F)} \langle w, x \rangle + \sum_{i \in V} \psi_i(x_i) 
=
\max_{w \in B(F)} \min_{x \in \mathbb{R}^n} \langle w, x \rangle + \sum_{i \in V} \psi_i(x_i) \,.
\end{align*}
Since each $\psi_i$ acts on a different coordinate we can write the inner minimization problem as
\[
\min_{x \in \mathbb{R}^n} \sum_{i \in V}\left( w_i x_i + \psi_i(x_i)\right)  = -\sum_{i \in V} \psi^*_i(-w_i)\,,
\]
where we applied the definition of the Fenchel dual. Furthermore by standard convex analysis~\cite{borwein2010convex, rockafellar1970convex}, as  $\psi_i'$ ranges from $-\infty$ to $\infty$ for each $i$ we have that $(\psi^*_i)'(-w_i) = x_i$, and similarly $\psi_i'(x_i) = -w_i$.

Thus we can equivalently write (\ref{eq:parametric-optimization}) as
\[
\max_{w \in B(F)} - \sum_{i \in V} \psi_i^*(-w_i)\,.
\]
By the previous observation, the optima are thus related via $(\psi^*_i)'(-w^*_i) = x^*_i$, and similarly $\psi_i'(x^*_i) = -w^*_i$.
\end{proof}

\section{Parametric $s$-$t$ Cuts}\label{sec:parm-st-cut}

In this section we show how to solve the parametric minimum cut problem by efficiently using a maximum flow oracle. In Section~\ref{sec:iterative} we show how to convert the solution obtained by this combinatorial routine to a nearly-optimal solution to a related optimization problem on the submodular base polytope of the corresponding cut function.

In the parametric min $s,t$-cut problem, we are given a directed network $G=(V,E)$ with two distinguished vertices: a \emph{source} $s\in V$, and a \emph{sink} $t\in V$, $s\neq t$. 
The capacities of individual edges of~$G$ are
nonnegative functions of a real parameter $\lambda$ in some possibly infinite domain $\dom\subseteq \mathbb{R}$ (as opposed to constants in the classical setting of min $s,t$-cut).
Following~\cite{GalloGT89}, we assume that the capacities of edges $sv\in E$ are
nondecreasing in $\lambda$ and the capacities of edges $vt\in E$ are nonincreasing in $\lambda$.
The capacities of all other edges of $G$ are constant.

We denote by $c_\lambda(uv):\dom\to\mathbb{R}$ the capacity function of an edge $uv\in E$.
Moreover, we assume that these edge capacity functions can be evaluated for arbitrary $\lambda$ in constant time.
\iffalse
, (2) the capacity functions can be efficiently combined, i.e., given two such functions $f_1,f_2:\dom\to\mathbb{R}$, in
$O(1)$ time we can compute a representation of $f'=f_1+f_2$ such that one can evaluate $f'$
in constant time, and later also possibly combine it with other functions without ever making the evaluation time super-constant.
For example, constant-degree polynomials form an example class of such capacity functions.
\fi

Roughly speaking, the goal of the parametric min $s,t$-cut problems is
to compute a representation of min $s,t$-cut 
for all the possible parameters $\lambda$.
Before we precisely define what this means,
let us introduce some more notation and state some
useful properties of (parametric) min-cuts. 

\newcommand{\ccap}{\ensuremath{\text{cap}}}

Denote by $\ccap(G)$ the capacity of a min $s,t$-cut in $G$.
Let $G_{\lambda'}$ be the graph with all the parameterized capacities replaced with the corresponding values for $\lambda=\lambda'$.
For any $S$, $s\in S\subseteq V\setminus\{t\}$, let $c_{\lambda}(S)$ be the capacity function of $S$, i.e., the sum of capacity functions $c_\lambda(uv)$ through all $uv$ with $u\in S$ and $v\in V\setminus S$.

\begin{lemma}[\cite{ford-fulkerson}]\label{l:minimal-min-cut}
For any $G$, there exists a unique \emph{minimal} minimum $s,t$-cut $(S,T)$ with $|S|$ smallest possible, such that for any min $s,t$-cut $(S',T')$ of~$G$ we have $S\subseteq S'$. Given any maximum $s,t$-flow $f$ in $G$, such a cut can be computed from $f$ in $O(m)$ time.
\end{lemma}
\begin{proof}Let $G_f$ be the residual network associated with flow $f$. We let $S$ be the set of
vertices reachable from $s$ in $G_f$ (via edges with positive capacity).
As proven by Ford and Fulkerson
~\cite[Theorem 5.5]{ford-fulkerson},~$S$ defined this way does not depend on the chosen maximum flow $f$, and $S\subseteq S'$ holds. Clearly, given~$f$,~$S$ can be found using any graph search algorithm.
\end{proof}

Ford and Fulkerson~{\cite[Corollary~5.4]{ford-fulkerson}} showed that for any two min $s,t$-cuts $(S_1,T_1),(S_2,T_2)$ of $G$, $(S_1\cap S_2,T_1\cup T_2)$ is also a min $s,t$-cut of $G$.
Gallo et al.~\cite[Lemma~2.8]{GalloGT89} gave the following generalization of this property to parametric min $s,t$-cuts.

\begin{lemma}[\cite{GalloGT89}]\label{l:gallo}
Let $\lambda_1\leq \lambda_2$.
For $i=1,2$, let $(S_{\lambda_i},T_{\lambda_i})$ be some min $s,t$-cut in $G_{\lambda_i}$. Then $(S_{\lambda_1}\cap S_{\lambda_2},T_{\lambda_1}\cup T_{\lambda_2})$ is a min $s,t$-cut in $G_{\lambda_1}$.
\end{lemma}
Our algorithm will use the following crucial property of parametric minimal min $s,t$-cuts.
\begin{lemma}\label{l:laminar}
Let $\lambda_1\leq \lambda_2$. For $i=1,2$, let $(S_{\lambda_i},T_{\lambda_i})$ be the unique minimal min $s,t$-cut in $G_{\lambda_i}$. Then $S_{\lambda_1}\subseteq S_{\lambda_2}$.
\end{lemma}
\begin{proof}
The uniqueness of $S_{\lambda_1}$ and $S_{\lambda_2}$ follows by Lemma~\ref{l:minimal-min-cut} applied to $G_{\lambda_1}$ and $G_{\lambda_2}$, respectively.
By Lemma~\ref{l:gallo}, $(S_{\lambda_1}\cap S_{\lambda_2},T_{\lambda_1}\cup T_{\lambda_2})$ is a min $s,t$-cut in $G_{\lambda_1}$.
By Lemma~\ref{l:minimal-min-cut}, we have $S_{\lambda_1}\subseteq S_{\lambda_1}\cap S_{\lambda_2}$.
It follows that $S_{\lambda_1}\subseteq S_{\lambda_2}.$
\end{proof}

Now given Lemma~\ref{l:laminar}, we can formally state our goal in this section, which is to compute
a parametric min $s,t$-cut defined as follows.
Let $\lmin\in \dom$ be such that
the minimal min $s,t$-cuts of $G_{\lmin}$ and
$G_{\lambda'}$ are equal
for all $\lambda'\in \dom$, $\lambda'<\lmin$.
Similarly, let $\lmax\in\dom$ be such
that the minimal min $s,t$-cuts of $G_{\lmax}$ and~$G_{\lambda'}$ are equal for all $\lambda'\in\dom$ with $\lambda'>\lmax$.
We will consider $\lmin$ and $\lmax$ additional inputs
to our problem.

For simplicity, in the remaining part of this section we denote by $S_\lambda$ and $T_\lambda$ the $s$-side and the $t$-side (resp.)
of \emph{the minimal} min-$s,t$-cut of $G_\lambda$.

\begin{defn}[Parametric min $s,t$-cut]\label{d:pmstc}
Let $\Lambda=\{\lambda_1,\ldots,\lambda_k\}\subseteq \dom$, where $k\leq n-1$ and\linebreak
${\lmin<\lambda_1<\ldots<\lambda_k\leq\lmax}$.
Let $\lambda_0=\lmin$.
Let ${\emb:V\to\Lambda\cup\{\lmin,\infty\}}$ be such
that $\emb(s)=\lmin$ and $\emb(t)=\infty$.
Let $S(z)=\{{v\in V}:\emb(v)\leq z\}$.
A pair $(\Lambda,\emb)$ is a parametric min
$s,t$-cut of $G$ if:
\begin{enumerate}%
    \item For $i=0,\ldots,k-1$, $S(\lambda_i)$ is a minimal min $s,t$-cut of $G_{\lambda'}$ 
    for all $\lambda'\in [\lambda_i,\lambda_{i+1})\cap \dom$.
    \item $S(\lambda_k)$ is a minimal min $s,t$-cut of $G_{\lmax}$.
    \item For $i=0,\ldots,k-1$, $S(\lambda_i)\subsetneq S(\lambda_{i+1})$.
\end{enumerate}
\end{defn}

It will also prove useful to define an approximate version of parametric min $s,t$-cut.
\begin{defn}[$\eps$-approximate parametric min $s,t$-cut]\label{d:eps-pmstc}
Let $\Lambda$, $\emb$, and $S:\dom\to 2^V$ be
as in Definition~\ref{d:pmstc}.
A pair $(\Lambda,\emb)$ is called an $\eps$-approximate parametric min
$s,t$-cut of $G$ if:
\begin{enumerate}%
    \item For $i=0,\ldots,k-1$, $S(\lambda_i)$ is a minimal min $s,t$-cut of $G_{\lambda'}$ 
    for all $\lambda'\in [\lambda_i,\lambda_{i+1}-\eps)\cap \dom$.
    \item $S(\lambda_k)$ is a minimal min $s,t$-cut of $G_{\lmax}$.
    \item For $i=0,\ldots,k-1$, $S(\lambda_i)\subsetneq S(\lambda_{i+1})$.
\end{enumerate}
\end{defn}

\begin{lemma}
Let $(\Lambda,\emb)$ be the parametric min $s,t$-cut of $G$. 
Let $(\Lambda_\eps,\emb_\eps)$ be an $\eps$-approximate parametric min $s,t$-cut of $G$.
Then for all $v\in V$, $\emb(v)\leq \emb_\eps(v)\leq \emb(v)+\eps$.
\end{lemma}
\begin{proof}
Let $S(z)=\{v\in V:\emb(v)\leq z\}$, and $S_\eps(z)=\{v\in V:\emb_\eps(v)\leq z\}$.
First of all, $\emb(v)=\infty$ if and only if $\emb_\eps(v)=\infty$.
This is because each of those is equivalent to $v\notin S_{\lmax}$.
In this case the lemma holds trivially.

So in the following let us assume that $\emb(v)$ and $\emb_\eps(v)$ are both finite.
We first prove $\emb_\eps(v)\geq \emb(v)$.
If $\emb(v)=\lmin$ then this follows by $\emb_\eps(v)\geq \lmin$.
So suppose $\emb(v)=\lambda$ for some $\lambda\in \Lambda$.
Then by item~(1) of Definition~\ref{d:pmstc}, for any $\lambda'<\lambda$, $S(\lambda')$ is a minimal min $s,t$-cut of $G_{\lambda'}$ and $v\notin S(\lambda')$.
If we had $\emb_\eps(v)<\emb(v)$, then $S_\eps(\emb_\eps(v))$ would
be a minimal min $s,t$-cut of $G_{\emb_\eps(v)}$ such that $v\in S_{\emb_\eps(v)}$ and $\emb_\eps(v)<\lambda$, a contradiction.

Now let us prove $\emb_\eps(v)\leq \emb(v)+\eps$. To this end, suppose $\emb_\eps(v)>\emb(v)+\eps$. If $\emb_\eps(v)=\lmin$,
then we have $\lmin>\emb(v)+\eps\geq \lmin+\eps$, a clear contradiction.
So let us assume that $\emb_\eps(v)\in \Lambda_\eps$
and let $\lambda^*$ be the element preceding $\emb_\eps(v)$ in $\Lambda_\eps$,
or $\lambda^*=\lmin$ if no such element exists.
We have $v\notin S_{\lambda^*}$ and $S_{\lambda^*}$ is a minimal min $s,t$-cut
in $G_{\lambda'}$ for $\lambda'=\lambda^*$ and all $\lambda'\in [\lambda^*,\emb_\eps(v)-\eps)$.
As a result, for any $\lambda''<\emb_\eps(v)-\eps$, 
the minimal min $s,t$-cut of $G_{\lambda''}$ does not contain $v$ in the $s$-side.
But $\emb(v)<\emb_\eps(v)-\eps$, $v\in S(\emb(v))$, and $S(\emb(v))$ is a minimal
min $s,t$-cut of $G_{\emb(v)}$, a contradiction.
\end{proof}

Our main result in this section is the following theorem.
\aprxpmc*

The rest of this section is devoted to proving
Theorem~\ref{t:aprx-pmc}.
For a connected subset $X\subseteq V(G)$, $\{s,t\}\not\subseteq X$, let $G/X$ denote $G$ after merging the vertex set $X$ into a single vertex. If the contracted vertex set $X$ contains $s$ ($t$), then the resulting vertex inherits the identity of $s$ ($t$, resp.).

\begin{lemma}\label{l:contract}
Let $\lambda$ be arbitrary and let $(S_{\lambda},T_{\lambda})$ be the minimal min $s,t$-cut in~$G_\lambda$. Then:
\begin{enumerate}
    \item For any $\lambda'\geq \lambda$,
    $\ccap(G_{\lambda'})=\ccap(G_{\lambda'}/S_{\lambda})$.
    \item For any $\lambda'\leq \lambda$,
    $\ccap(G_{\lambda'})=\ccap(G_{\lambda'}/T_{\lambda})$.
\end{enumerate}
\end{lemma}
\begin{proof}
We only prove item 1, as item 2 is analogous.
Since merging vertices is equivalent to connecting them with infinite capacity edges,
it cannot decrease the min $s,t$-cut capacity, i.e.,
$\ccap(G_{\lambda'})\leq \ccap(G_{\lambda'}/S_{\lambda})$.
On the other hand, by Lemma~\ref{l:laminar}, the minimal $s,t$ min-cut $(S_{\lambda'},T_{\lambda'})$ in $G_{\lambda'}$ satisfies $S_\lambda\subseteq S_{\lambda'}$.
Hence, the capacity of the $s,t$-cut $(S_{\lambda'}/S_{\lambda},T_{\lambda'})$ in
$G_{\lambda'}/S_{\lambda}$ is the same as
the capacity $\ccap(G_{\lambda'})$ of $(S_{\lambda'},T_{\lambda'})$
in~$G_{\lambda'}$.
Consequently, $\ccap(G_{\lambda'})\geq \ccap(G_{\lambda'}/S_{\lambda})$.
\end{proof}
\begin{remark}\label{rem:minor}
If $(S_\lambda,T_\lambda)$ is a minimal min $s,t$-cut in $G_\lambda$, then $G[S_\lambda]$
is connected by construction (Lemma~\ref{l:minimal-min-cut}).
However, $G[T_\lambda]$ might in general consist of several connected components if $G_\lambda$ contains zero-capacity edges. In that case, we can still obtain $G_{\lambda'}/T_{\lambda}$ above using edge/vertex deletions and edge contractions.
Namely, we contract only the connected component $A$
of $T_{\lambda}$ that contains $t$.
For any other component $C_i$ ($i=1,\ldots,q$) of $T_{\lambda}$, its incoming edges start in $S_{\lambda}$ and all have capacity $0$ in $G_{\lambda}$, and thus also
in $G_{\lambda'}$ for $\lambda'<\lambda$.
Consequently, removing the vertices of $\bigcup_{i=1}^q C_i$ and subsequently contracting~$A$ has the same effect on $G_{\lambda'}$ as merging the entire $T_{\lambda}$,
i.e., $G_{\lambda'}/T_{\lambda}=G_{\lambda'}[V\setminus \bigcup_{i=1}^q C_i]/A$.
\end{remark}

We use a recursive ``divide-and-conquer'' algorithm. The input
to a recursive procedure \linebreak
$\textsc{ApxParametricMinCut}$
is a graph $G=(V,E)$ with $n$ vertices, $m$ edges, source $s$ and sink $t$, the parametric capacity function
$c_\lambda:E\to \dom\to \mathbb{R}$, and two parameters $\lmin,\lmax$
such that~$\eps$ evenly divides $\lmax-\lmin$.
The output of the procedure is an $\eps$-approximate
parametric min $s,t$-cut $(\{\lambda_1,\ldots,\lambda_k\},\emb)$
as in Definition~\ref{d:eps-pmstc}.
By Lemma~\ref{l:laminar}, $k\leq |V(G)|-1$.

\begin{algorithm}[h!]
\caption{Computing an $\eps$-approximate parametric min $s,t$-cut.}\label{alg:find_mincuts}
\begin{algorithmic}[1]
\STATE Let $s,t,\eps$ be globally defined.
\vspace{2mm}

\STATE {\bfseries function }{\textsc{ApxParametricMinCut}$(G=(V,E),c_\lambda:E\to \dom\to \mathbb{R},\lmin\in\dom,\lmax\in\dom)$}
\Indent
\IF{$|V|\leq 2$}
    \RETURN$(\emptyset,\{s\to \lmin,t\to \infty\})$
\ENDIF
\STATE For any $\lambda'\in\dom$, let $c_\lambda[\lambda=\lambda']$ the capacity function $E\to\mathbb{R}$ of $G_{\lambda'}$
\STATE $S_{\lmin}=\textsc{MinimalMinCut}(G,c_\lambda[\lambda=\lmin])$
\STATE $S_{\lmax}=\textsc{MinimalMinCut}(G,c_\lambda[\lambda=\lmax])$
\IF{$|S_{\lmin}|>|V|/2$}
    \RETURN $\textsc{ApxParametricMinCut}(\textsc{Contract}(G,c_\lambda,S_{\lmin}),\lmin,\lmax)$
\ENDIF
\IF{$|S_{\lmax}|<|V|/2$}
    \RETURN $\textsc{ApxParametricMinCut}(\textsc{Contract}(G,c_\lambda,V\setminus S_{\lmax}),\lmin,\lmax)$
\ENDIF
\STATE $(\lambda_1,\lambda_2):=(\lmin,\lmax)$
\WHILE{$\lambda_2-\lambda_1>\eps$}
    \STATE $\lambda':=\lambda_1+\lfloor(\lambda_2-\lambda_1)/2\eps\rfloor\cdot\eps$
    \STATE $S_{\lambda'}=\textsc{MinimalMinCut}(G,c_\lambda[\lambda=\lambda'])$
    \IF{$|S_{\lambda'}|\geq |V|/2$}
        \STATE $\lambda_2:=\lambda'$
    \ELSE
        \STATE $\lambda_1:=\lambda'$
    \ENDIF
\ENDWHILE
\STATE For $i=1,2$, $S_{\lambda_i}:=\textsc{MinimalMinCut}(G,c_{\lambda}[\lambda=\lambda_i])$
\STATE $(\Lambda_1,\tau_1)=\textsc{ApxParametricMinCut}(\textsc{Contract}(G,c_\lambda,V\setminus S_{\lambda_1}),\lmin,\lambda_1)$
\STATE $(\Lambda_2,\tau_2)=\textsc{ApxParametricMinCut}(\textsc{Contract}(G,c_\lambda, S_{\lambda_2}),\lambda_2,\lmax)$
\STATE $\Lambda:=\textbf{ if }|S_{\lambda_1}|=|S_{\lambda_2}|\textbf{ then }\Lambda_1\cup \Lambda_2\textbf{ else }\Lambda_1\cup \{\lambda_2\}\cup \Lambda_2$
\STATE $\tau:=\{v\in S_{\lambda_1}\to \tau_1(v),
v\in V\setminus S_{\lambda_2}\to\tau_2(v),
v\in S_{\lambda_2}\setminus S_{\lambda_1}\to \lambda_2\}$
\RETURN $(\Lambda,\tau)$
\EndIndent

\algstore{myalg}

\end{algorithmic}
\label{alg:main}
\end{algorithm}

\begin{algorithm}
\begin{algorithmic}[1]
\algrestore{myalg}

\vspace{0.2cm}
\STATE {\bfseries function }{\textsc{MinimalMinCut}$(G=(V,E),c:E\to\mathbb{R})$}
\Indent
    \STATE $f=\textsc{MaxFlow}(G,s,t,c)$
    \RETURN $\{v\in V: v\text{ reachable from } s\text{ in the residual network }G_f\}$
\EndIndent
\vspace{0.2cm}
\STATE {\bfseries function }{\textsc{Contract}$(G=(V,E),c_\lambda:E\to \dom\to \mathbb{R},X\subseteq V)$}\quad\COMMENT{$|X\cap\{s,t\}|=1$}
\Indent

\STATE $w^*:=\textbf{ if }s\in X\textbf{ then }s\textbf{ else }t$
\STATE $V':=V\setminus X\cup \{w^*\}$
\STATE $E':=\emptyset$
\STATE $c'_\lambda:=E'\to \dom\to\mathbb{R}$
\FOR{$uv\in E$}
    \STATE $u':=\textbf{ if }u\in X\textbf{ then }w^*\textbf{ else }u$
    \STATE $v':=\textbf{ if }v\in X\textbf{ then }w^*\textbf{ else }v$
    \IF{$(u',v')\neq (s,t)$}
        \IF{$u'v'\notin E'$}
            \STATE $E':=E'\cup \{u'v'\}$
            \STATE $c'_\lambda(u'v'):=c_\lambda(uv)$
        \ELSE 
            \STATE $c'_\lambda(u'v'):=c'_\lambda(u'v')+c_\lambda(uv)$\quad\COMMENT{We add functions here.}
        \ENDIF
    \ENDIF
\ENDFOR
\RETURN $(G'=(V',E'),c'_\lambda)$
\EndIndent

\end{algorithmic}
\end{algorithm}

The main idea of the procedure~\textsc{ApxParametricMinCut} is to find the (approximately) most balanced minimal $s,t$-cuts $S_{\lambda_1}$ and $S_{\lambda_2}$ and use them to reduce the problem size in the recursive calls significantly. Specifically, we
want to find such $\lambda_1\leq \lambda_2$ that
$|S_{\lambda_1}|\leq n/2$,
$|S_{\lambda_2}|\geq n/2$ and  $\lambda_2-\lambda_1=\eps$.

Suppose $n>2$ as otherwise the problem is trivial. First, we compute minimal min-cuts in $G_{\lmin}$ and $G_{\lmax}$.
This takes two max-flow runs, i.e., $\tmf(n,m)$ time, plus $O(m)$ time by
Lemma~\ref{l:minimal-min-cut}.

It might happen that $|S_{\lmin}|\leq |S_{\lmax}|<n/2$ or $n/2<|S_{\lmin}|\leq |S_{\lmax}|$.
In these special cases we can immediately
reduce the vertex set by a factor of at least two by contracting $T_{\lmax}$
or $S_{\lmin}$ respectively,
and recurse on the reduced graph.
By Lemma~\ref{l:contract} and the definition of $\lmin,\lmax$ this reduction does not influence the structure of parametric cuts.

So suppose $|S_{\lmin}|\leq n/2$
and $|S_{\lmax}|\geq n/2$.
Set $\lambda_1=\lmin$ and $\lambda_2=\lmax$. So we have 
$|S_{\lambda_1}|\leq n/2$
and $|S_{\lambda_2}|\geq n/2$
initially.
We maintain this invariant and gradually
shrink the interval $[\lambda_1,\lambda_2]$ until its length gets precisely $\eps$ in a binary search-like way. We repeatedly try the pivot ${\lambda'=\lambda_1+\lfloor(\lambda_2-\lambda_1)/2\eps\rfloor\cdot\eps}$
and compute $S_{\lambda'}$. If
$|S_{\lambda'}|\geq n/2$, we set $\lambda_2=\lambda'$,
and otherwise we set $\lambda_1=\lambda'$.
Note that $\lambda_2-\lambda_1$
remains an integer multiple of $\eps$ at all times.
The whole process costs $O(\log[(\lmax-\lmin)/\eps])=O(\log(R/\eps))$
max-flow executions.

Let $G_1=G/T_{\lambda_1}$ and $G_2=G/S_{\lambda_2}$. 
Note that $G_1,G_2$ may contain parallel edges
or a direct $st$ edge as a result of contraction.
Hence, these graphs are first
preprocessed by (1) removing self-loops and direct $st$ edges,
(2) merging parallel edges by summing their cost functions.
The contraction and preprocessing
is performed using the procedure \textsc{Contract}.
Note that none of these preprocessing steps
change the minimal cuts of $G_{1,\lambda}$
or $G_{2,\lambda}$ for any $\lambda$:
the direct $st$ edges cross \emph{all} $s,t$-cuts.

Next, we recursively compute $\eps$-approximate
parametric min $s,t$-cut in graphs
$G_1=G/T_{\lambda_1}$ and $G_2=G/S_{\lambda_2}$,
The recursive call on $G_1$ is made
with $(\lmin,\lmax)$ set to $(\lmin,\lambda_1)$, whereas the recursive call on $G_2$
uses $(\lmin,\lmax):=(\lambda_2,\lmax)$.
Note that indeed we have $G_{1,\lambda_1}=G_{1,\lambda'}$
for $\lambda'>\lambda_1$ as required
since the $t$-side of the minimal
min $s,t$-cut in $G_{1,\lambda_1}$
contains only $t$.
Similarly, $G_{2,\lambda_2}=G_{2,\lambda'}$
for all $\lambda'<\lambda_2$.

Let $(\Lambda_1,\emb_1)=(\{\lambda_{1,1},\ldots,\lambda_{1,a}\},\emb_1)$ and
$(\Lambda_2,\emb_2)=(\{\lambda_{2,1},\ldots,\lambda_{2,b}\},\emb_2)$
be the returned $\eps$-approximate parametric min
$s,t$-cuts of $G_1$  and $G_2$ respectively.
We return $(\Lambda,\emb)$
as the $\eps$-approximate parametric min-$s,t$-cut
of $G$, where
\begin{align*}
    \Lambda=\begin{cases}\Lambda_1\cup\Lambda_2 & \text{if } |S_{\lambda_1}|=|S_{\lambda_2}|=n/2,\\
    \Lambda_1\cup\{\lambda_2\}\cup \Lambda_2&\text{otherwise.}
    \end{cases} & \hspace{1cm}\emb(v)=\begin{cases}
        \emb_1(v) & \text{if } v\in S_{\lambda_1},\\
        \emb_2(v) & \text{if } v\in T_{\lambda_2},\\
        \lambda_2 & \text{otherwise.}\end{cases}
\end{align*}

Let us now prove the correctness of this algorithm.
We proceed by induction on $n$.
For $n\leq 2$ this is trivial, so suppose $n>3$
and that recursive calls are made.
Clearly, $\lambda_{1,a}\leq \lambda_1<\lambda_2<\lambda_{2,1}$.

Item~(3) of Definition~\ref{d:eps-pmstc} follows easily by induction
and the definition of $\lambda_1,\lambda_2$.
Let $\Lambda=\{\lambda_1',\ldots,\lambda_k'\}$.
That $S(\lambda_i')=\{v\in V:\emb(v)\leq \lambda_i'\}$ is a minimal min $s,t$-cut of $G_{\lambda_i'}$ for all $\lambda_i'\in \Lambda$ (i.e., item~(2) of Definition~\ref{d:eps-pmstc}) follows directly by Lemma~\ref{l:contract}
the definitions of $\lambda_1,\lambda_2$.

Now consider item~(1) of Definition~\ref{d:eps-pmstc}.
For some $j<k$ we have $\lambda_j'=\lambda_{1,a}$.
For all $i=0,\ldots,k-1$, $i\neq j$, item~(1), i.e.,
that $S_{\lambda_i'}$ is a minimal min $s,t$-cut for all
$\lambda'\in [\lambda_i',\lambda_{i+1}')$, follows 
directly inductively.

If $|S_{\lambda_1}|=|S_{\lambda_2}|=n/2$,
then $S_{\lambda_{1,a}}=S_{\lambda_2}$. By induction
it follows that $S_{\lambda_{1,a}}$ is a minimal min $s,t$-cut
of $G_{\lambda'}$ for all $\lambda'\in [\lambda_2,\lambda_{2,1}-\eps)$,
and thus also for all $\lambda'\in [\lambda_{1,a},\lambda_{2,1}-\eps)=[\lambda_j',\lambda_{j+1}'-\eps)$.

If, on the other hand, $S_{\lambda_1}\subsetneq S_{\lambda_2}$,
then $\lambda_{j+1}'=\lambda_2$.
Since $S_{\lambda_{1,a}}=S_{\lambda_1}$, $S_{\lambda_{1,a}}$
is indeed a minimal min $s,t$-cut for all $\lambda'\in [\lambda_{1,a},\lambda_2-\eps)=[\lambda_j',\lambda_{j+1}'-\eps)$ as $\lambda_2-\eps=\lambda_1$.

 Note that the input graph of each of the recursive calls has at most $n/2+1$ vertices.
Moreover, by merging the parallel edges (and summing their costs) after the contraction
we can guarantee that $|E(G_1)|+|E(G_2)|\leq |E(G)|+n/2$.
Indeed, observe that the only edges
of $G_2$ that can also appear in~$G_1$
are those incident to $s$ in $G_2$,
and there are at most $|T_{\lambda_2}|\leq n/2$ of them.

There is one important technical detail
here: even though the individual functions $c_\lambda(uv)$ (for the edges $uv$ of the original input graph $G$) can be evaluated in constant
time, after summing $k$ of such functions in the process this
cost can be as much as $\Theta(k)$. We now argue
that this cannot happen in our case due to preprocessing $G_1$ and $G_2$, and the evaluation cost is $O(1)$ for all edges in all recursive calls. More concretely, one can show that each edge capacity function in a recursive call can be expressed as the sum of at most one original capacity function $c_\lambda(xy)$ and a real number. Indeed, suppose this is the case for some call with input $G$. Then, each edge
$uv$ of $G_1$ either (1) is contained in $G$ and has not resulted from merging some parallel edges after contraction, (2) has $u\notin \{s,t\}$ and $v=t$ and resulted from merging edges $uz_1,\ldots,uz_l$
such that $z_1,\ldots,z_l\in T_{\lambda_1}$.
The former case is trivial.
In the latter case, for at least $l-1$ of these $z_i$ we have $z_i\neq t$, so $c_\lambda(uz_i)$ is a constant function.
Fot at most one $z_j$ is of the form $c_\lambda(xy)+\Delta$ for some original
capacity function $c_\lambda(xy)$ and $\Delta\in\mathbb{R}$.
We conclude that the capacity function of $uv$
in $G_1$ is of the same form and equals
$c_\lambda(xy)+\Delta'$, where $\Delta'=\Delta+\sum_{i\neq j}c_\lambda(uz_i)\in\mathbb{R}$.
The proof for $G_2$ is analogous.

Now let us analyze the running time of the algorithm.
One can easily inductively prove that:
\begin{itemize}
    \item Each graph at the $i$-th level of the recursion tree has less than $n/2^i+2$ vertices; hence, there
    are no more than $\log_2{n}+1$ levels in the tree.
    \item The sum $n_i$ of numbers of vertices through all the graphs at the $i$-th level
    is less than\linebreak $2^i(n/2^i+2)\leq n+2^{i+1}\leq 3n$.
    \item Since, the sum $m_i$ of numbers of edges
    in graphs at level $i>0$
    satisfies $m_i\leq m_{i-1}+n_{i-1}/2\leq m_{i-1}+3n/2$,
    we have $m_i\leq m+3in/2=O(m+n\log{n})$.
\end{itemize}
By the above, and since the function $\tmf$ is convex, 
we conclude the total time cost at the $i$-th level
is $O\left(\tmf(n,m+n\log{n})\log(R/\eps)\right)$.
Recall that there are $O(\log{n})$ levels and therefore the total time is $O\left(\tmf(n,m+n\log{n})\log(R/\eps)\log{n}\right)$.

\subsection{Exact Parametric Min $s,t$-Cut}
In this section we show how Theorem~\ref{t:aprx-pmc} implies
new bounds on computing \emph{exact}
parametric min $s,t$-cuts in a few interesting settings.

\paragraph{Integer Polynomial Costs.}
Suppose all the parametric costs $c_\lambda(uv)$ are
of the form $c_\lambda(uv)=Q_{uv}(\lambda)$,
where each $Q_{uv}$ is a (possibly different) constant-degree polynomial
with integer coefficients bounded
in the absolute value by an integer $U>0$
and take nonnegative values on $\dom$.
Recall $Q_{uv}$ can have a positive degree only if $u=s$ or $v=t$.
Moreover, if $u=s$, then $Q_{uv}$
is increasing, whereas when $v=t$,
then $Q_{uv}$ is decreasing.

Observe that the parametric capacity $c_\lambda(S)$ of any $S$, $s\in S\subseteq V\setminus \{t\}$ is
a constant-degree polynomial
with integer coefficients bounded by $nU$ in absolute value.
The same applies to a difference
polynomial 
$c_\lambda(S)-c_\lambda(S')$ for
any two such sets $S,S'$.

It is known that for a constant-degree polynomials $Q$ with integer coefficients bounded by $W$:
\begin{itemize}
    \item The roots of $Q$ are of absolute value $O(\poly(W))$ (e.g.,~\cite{yap2000fundamental}).
    \item Any two distinct roots of $Q$
    are at least $\Omega(1/\poly(W))$
    apart.~\cite{mahler1964}
\end{itemize}
This means that by setting $\lmin=-R/2$ and $\lmax=R/2$ (or slightly less aggressively, if e.g., $R/2\notin \dom$) for a sufficiently 
large even integer $R=O(\poly{nU})$ such that $R/2$ exceeds the maximum possible absolute value of a root of any polynomial
of the form $c_\lambda(S)-c_\lambda(S')$, we
will
indeed have $G_{\lmin}=G_{\lambda'}$
for all $\lambda'<\lmin$
and
$G_{\lmax}=G_{\lambda'}$
for all $\lambda'>\lmax$.

Moreover, assume we compute an $\eps$-approximate
parametric min $s,t$-cut $(\Lambda,\tau)$, where
\linebreak
$\Lambda=\{\lambda_1,\ldots,\lambda_k\}$. Suppose for some $i$ 
there exists $\lambda'$, $\max(\lambda_i,\lambda_{i+1}-2\eps)\leq \lambda'<\lambda_{i+1}$,
such that the minimal $s,t$-cut $S_{\lambda'}$ in $G_{\lambda'}$ satisfies $S_{\lambda_i}\subsetneq S_{\lambda'}\subsetneq S_{\lambda_{i+1}}$.
Let $\lambda_i^*=\max(\lambda_i,\lambda_{i+1}-2\eps)$.
Note that $S_{\lambda_i}=S_{\lambda_i^*}$
by Definition~\ref{d:eps-pmstc}.
Since $S_{\lambda'}$ is minimal,
$c_\lambda(S_{\lambda_i})(\lambda_i^*)- c_\lambda(S_{\lambda'})(\lambda_i^*)\leq 0$ 
and
${c_\lambda(S_{\lambda_i})(\lambda')-c_\lambda(S_{\lambda'})(\lambda')>0}$.
So the polynomial $c_\lambda(S_{\lambda_i})-c_\lambda(S_{\lambda'})$ is non-zero and has a root in the interval
$[\lambda_i^*,\lambda')$.
Similarly one can prove
that the polynomial $c_{\lambda}(S_{\lambda'})-c_{\lambda}(S_{\lambda_{i+1}})$ is non-zero and has a root
in the interval $[\lambda',\lambda_{i+1})$.
We conclude that the product polynomial
$[c_\lambda(S_{\lambda_i})-c_\lambda(S_{\lambda'})]\cdot [c_{\lambda}(S_{\lambda'})-c_{\lambda}(S_{\lambda_{i+1}})]$ is non-zero, has constant degree, has integer coefficients of order $O(\poly{nU})$, and has
two distinct roots in the interval
$[\lambda_i^*,\lambda_{i+1})$,
i.e., less than $2\eps$ apart.
Therefore, if we set $\eps$ so that
$1/\eps$ is a sufficiently large integer but still polynomial in $nU$,
the assumption $S_{\lambda_i}\subsetneq S_{\lambda'}\subsetneq S_{\lambda_{i+1}}$ leads to a contradiction.
As a result, for
all such $\lambda'$, $S_{\lambda'}$ equals either
$S_{\lambda'}$ or $S_{\lambda''}$.

In other words, computing an $\eps$-approximate
min $s,t$-cut $(\Lambda,\tau)$, where $\Lambda=\{\lambda_1,\ldots,\lambda_k\}$, gives as the structure
of all possible minimal min $s,t$-cuts  $S_{\lambda}$ in the following sense.
Suppose $\Lambda^*=\{\lambda_1^*,\ldots,\lambda_l^*\}$ is an \emph{exact} parametric min $s,t$-cut.
Then $k=l$ and $S(\lambda_i)=\{v\in V: \tau(v)\leq \lambda_i)=S_{\lambda_i^*}$
for all $i=1,\ldots,k$.
To compute $\lambda_i^*$, it is hence enough
to find the \emph{unique} $\lambda^*_i\in (\lambda_{i-1},\lambda_i]$
such that $c_\lambda(S(\lambda_{i-1}))(\lambda^*_i)=c_\lambda(S(\lambda_i))(\lambda^*_i)$ which boils
down to solving a polynomial equation
of constant degree.
It is well known that such equations can be solved
exactly in constant time for degrees at most~$4$.

Observe that if we run our $\eps$-approximate parametric min $s,t$-cut algorithm with $\lmin,\lmax,\eps$ set as described above, maximum flow
is always invoked on some minor $H$
of $G_{\lambda}$ for $\lambda$
that is an integer multiple of $\eps$. Since, $1/\eps$ is an integer, by multiplying edge
costs in $H$ by $1/\eps$,
we only need a maximum flow
algorithm that can handle integer
edge capacities of order $O(\poly(nU))$. By plugging
in the best-known algorithms
for computing max flow with integral capacities, we obtain the following.

\begin{thm}
Let $G$ be a graph whose parameterized capacities are constant-degree polynomials with integer coefficients in $[-U,U]$.
The structure of parametric min $s,t$-cut on $G$ can be computed in:
\begin{itemize}
    \item $O(m\cdot\min(m^{1/2},n^{2/3})\cdot \polylog\{n,U\})$
    time using a combinatorial
algorithm~\cite{goldberg1998beyond},
    \item $O((m+n^{1.5})\cdot \polylog\{n,U\})$ time using
    the algorithm of~\cite{flow-dense},
    \item $O(m^{1.497}\polylog\{n,U\})$
    time using the algorithm of~\cite{flow-sparse}.
\end{itemize}
The cut function can be found exacly in additional 
$O(n)$ time if the degrees of capacity polynomials
are at most $4$.
\end{thm}

\paragraph{Discrete Domains.}
Let us now consider the case when $\dom$ is discrete and has $\ell$ elements.
Suppose all parametric costs are arbitrary functions meeting the requirements of the parametric min $s,t$-cut problem.
Then, we can make the
$\eps$-approximate algorithm exact
by employing the following simple modifications.
We start with $\lmin=\min\dom$
and $\lmax=\max\dom$.
In the binary-search like step,
we always choose the middle
element of $\dom\cap [\lambda_1,\lambda_2]$ as the next pivot. This way, all the max-flow
computations are performed
on minors of $G_\lambda$ for $\lambda\in \dom$.

\begin{thm}
Let $G$ be a graph with arbitrary parameterized capacities $\dom\to\mathbb{R}$ for a discrete
domain $\dom\subseteq\mathbb{R}$, where $\ell=|\dom|$. Let $\tmf(n',m')$ be defined
as in Theorem~\ref{t:aprx-pmc}.
Then exact parametric min $s,t$-cut on $G$ can be 
computed in $O(\tmf(n,m\log{n})\cdot\log{\ell}\cdot\log{n})$ time.
\end{thm}

\paragraph{Planar Graphs.} By Remark~\ref{rem:minor}, all the max-flow computations in the algorithm of Theorem~\ref{t:aprx-pmc} are performed
on minors of $G$. As a result, if the input
graph $G$ is planar, we can use state-of-the-art
planar max $s,t$-flow algorithms to obtain
better bounds on the parametric min-$s,t$-cut algorithms
on planar graphs. Since maximum $s,t$-flow
for planar graphs can be computed in $O(n\log{n})$
time even for real capacities~\cite{BorradaileK09, Erickson10}, planar parametric min $s,t$-cut can be solved exactly:
\begin{itemize}
    \item in $O(n\polylog\{n,U\})$ time when parameterized capacities are constant degree polynomials with integer coefficients in $[-U,U]$,
    \item in $O(n\log^3{n}\log{\ell})$ time for discrete domains $\dom\subseteq\mathbb{R}$
    of size $\ell$.
\end{itemize}

What may be surprising, our reduction is powerful enough to obtain
an interesting subquadratic \emph{strongly polynomial} exact algorithm
computing parametric min $s,t$-cut in a planar graph with
capacity functions that are arbitrary polynomials of degree no more than $4$
and real coefficients.

We now sketch this algorithm. It is based on the \emph{parametric search} technique~\cite{Megiddo83} (see also~\cite{AgarwalST94}). Suppose we want to solve some decision
problem $\mathcal{P}(\alpha)$, where $\alpha\in\mathbb{R}$, such that if $\mathcal{P}(\alpha_0)$ is
a yes instance, then $\mathcal{P}(\alpha')$ for all $\alpha''<\alpha_0$
is also a yes instance. We wish to find the maximum $\alpha^*$ for
which $\mathcal{P}(\alpha^*)$ is a yes instance.
An example problem $\mathcal{P}(\alpha)$ could be \emph{``does an $s,t$-flow
of value $\alpha$ exist in $G$?''}.
Then, $\alpha^*$ clearly equals the maximum flow in $G$.

\newcommand{\Ot}{\tilde{O}}

Suppose we have an efficient strongly polynomial algorithm solving the decision
problem. Then, in practice one could find $\alpha^*$ via binary search given
some initial interval containing $\alpha^*$;
however, in general this would not lead to an exact algorithm for real values
of $\alpha^*$. Parametric search is a technique
for converting a strongly polynomial \emph{parallel} decision algorithm
into a \emph{sequential or parallel} strongly polynomial \emph{optimization} algorithm as explained above. The only requirement to keep in mind is that the decision algorithm is governed
by comparisons, each of which amounts to testing the sign of some
low-degree (say, no more than~$4$) polynomial in $\alpha$.
Specifically, suppose we have a parallel decision algorithm $\mathcal{A}$ that uses $W_{\mathcal{A}}$ work and $D_{\mathcal{A}}$ depth, and also a (possibly the same) another parallel decision algorithm $\mathcal{B}$ with
$W_{\mathcal{B}}$ work and $D_{\mathcal{B}}$ depth.
Suppose for simplicity all these quantities are polynomial in $n$.
Then, parametric search yields a strongly-polynomial optimization
algorithm computing $\alpha^*$ in $\Ot(D_{\mathcal{A}}\cdot W_{\mathcal{B}}+W_{\mathcal{A}})$ work
and $\Ot(D_{\mathcal{A}}\cdot D_{\mathcal{B}})$ depth.

Now back to planar graphs. We will use parametric search in a nested way.
First of all, we will need a decent parallel max flow algorithm for planar
graphs. It is well-known~(e.g.,~\cite{Erickson10}) that the decision variant of the max $s,t$-flow
problem on planar graphs is reducible to negative cycle detection
in the dual graph (which is also planar).
There exists a parallel negative cycle detection algorithm
on planar graphs with $\Ot(n+n^{3/2}/d^3)$ work and $\Ot(d)$ depth
for any $d\geq 1$~\cite{parallel-apsp}.
Hence, by using that algorithm as both $\mathcal{A}$ 
(for $d=D^{3/7}$) 
and $\mathcal{B}$ (for $d=D^{4/7}$), where $D$ is a parameter,
we have $W_{\mathcal{A}}=\Ot(n+n^{3/2}/D^{9/7})$, $D_{\mathcal{A}}=\Ot(D^{3/7})$,
$W_{\mathcal{B}}=\Ot(n+n^{3/2}/D^{12/7})$, $D_{\mathcal{B}}=\Ot(D^{4/7})$.
So parametric search yields a strongly-polynomial parallel
max-flow algorithm for planar graphs with work $\Ot(n+n^{3/2}/D^{9/7})$ and depth $\Ot(D)$
for any $D\geq 1$.

Given a parallel algorithm for max flow in planar graphs, we can
use parametric search (instead of binary search) once again when computing the pair
$\lambda_1,\lambda_2$ in our recursive algorithm.
More specifically, we would like $\lambda_1$ to be the largest
such that $|S_{\lambda_1}|\leq n/2$,
whereas $\lambda_2$ to be the smallest such that $|S_{\lambda_2}|\geq n/2$.
It is easy to see that $\lambda_1$ and $\lambda_2$ are precisely neighboring
(or the same) breakpoints of the cut function, i.e., belong
to $\Lambda$ from Definition~\ref{d:pmstc}.
To actually compute $\lambda_1,\lambda_2$, we use parametric
search with $\mathcal{A}$ set to the obtained parallel
max-flow algorithm\footnote{Actually, it is computing max-flow followed by a graph search to determine the minimal min $s,t$-cut. However, this latter step
does not involve any comparisons on capacities, so its depth can be ignored.}, and $\mathcal{B}$ to the best known algorithm
that computes a minimum min $s,t$-cut in a planar graph,
i.e., a combination of the max-flow algorithm of~\cite{BorradaileK09, Erickson10},
and linear time graph search.
So, in the outer parametric search instance we have$W_{\mathcal{A}}=\Ot(n+n^{3/2}/D^{9/7})$, $D_{\mathcal{A}}=\Ot(D)$,
and $W_{\mathcal{B}}=\Ot(n)$.
Therefore, the obtained algorithm runs in $\Ot(n^{3/2}/D^{9/7}+Dn)$ sequential time.
By setting $D=n^{7/32}$, we obtain $\Ot(n^{1+7/32})=\Ot(n^{1.21875})$ time.

We stress that all the algorithms~\cite{BorradaileK09, Erickson10, parallel-apsp} used above proceed by only adding and comparing edge weights. Adding polynomials cannot increase their degrees, so indeed when these algorithms are run ``generically'' for some $\lambda$, the control flow depends only on signs of some small degree polynomials.

\begin{thm}
Let $G$ be a planar graph whose parameterized capacities are all polynomials of degree at most $4$ with real coefficients.
There exists a strongly polynomial algorithm computing parametric min $s,t$-cut in $G$ exactly in $\Ot(n^{1+7/32})$ time.
\end{thm}

\iffalse
There is one subtle technical detail though: note
that when contracting edges, we also ``combine''  costs
of functions parallel edges.
Assuming~$k$ edges parallel are merged, we cannot assume
their combined cost function can be evaluated faster
that in $\Theta(k)$ time. This overhead makes the
$O(\tmf(n,m)\polylog{n})$ bound not hold.
However, in some cases, e.g., when the cost
functions are constant-degree polynomials (with real coefficients),
the cost functions can be combined in constant time
so that evaluations take constant time as well.
\fi

\section{Removing  Assumptions on $F_i$}\label{sec:assumptions}
In this section we argue why the assumptions on the functions $F_i$ introduced in Section~\ref{sec:prelim} are valid without loss of generality. More precisely, we assumed that for all $i$, $F_i(\emptyset) = F_i(V_i) = 0$ and $F_i(S) \geq 0$ for all $S$.
Here we show that a simple preprocessing step can enforce all of these conditions.

Without changing the original problem we can shift each $F_i$ such that it evaluates to $0$ on $\emptyset$, by defining ${\overline{F}}_i(S) = F_i(S) - F_i(\emptyset)$. This only changes $F$ by a constant term without affecting the sets that minimize the parametric problem  (\ref{eq:parametric-def}).

For each $i$, we use Lemma~\ref{lem:opt-base-poly} to find a point $w_i \in B(\overline{F}_i)$. Using this point, we define $\overline{\overline{F}}_i(S) = \overline{F}_i(S) - w_i(S)$. Since by definition $w_i(V_i) = \overline{F}_i(V_i)$, we have that $\overline{\overline{F}}_i(V_i) = 0$. Also, we have $\overline{\overline{F}}_i(\emptyset) = \overline{F}_i(\emptyset)  = 0$. Finally, since $w_i(S) \leq \overline{F}_i(S)$ for all $S$, we also have $\overline{\overline{F}}_i(S) \geq 0$.

Now we can equivalently rewrite the parametric problem
\begin{align*}
F_\alpha(A) &= F(A) + \sum_{j \in A}   \psi_j'(\alpha) \\
&= \overline{F}(A) + \sum_{j \in A}   \psi_j'(\alpha) + \left( \sum_{i=1}^m F_i(\emptyset) \right) \\
&= \overline{\overline{F}}(A) + \sum_{j \in A} \left(  \psi_j'(\alpha) + \sum_{i=1}^\m w_i(j) \right) + \left( \sum_{i=1}^m F_i(\emptyset) \right)\,. \\
\end{align*}
Now we can solve the problem on $\overline{\overline{F}} = \sum_{i=1}^{\m} \overline{\overline{F}}_i$ with the parametric penalties $\overline{\overline{\psi}}'_j (\alpha) = \psi'_j(\alpha) + \sum_{i=1}^\m w_i(j)$, which maintain the validity of Assumption~\ref{defn:conditions}.

To compute a point in the base polytope of a submodular function we use the following folklore lemma, which shows that the running time of our initialization procedure is $O(\sum_{i=1}^\m \vert V_i \vert \cdot \evo_i)$:
\begin{lem}[\cite{fujishige1980lexicographically}]\label{lem:opt-base-poly}
Let $F : 2^V \rightarrow \mathbb{Z}$ be a submodular set function, with $F(\emptyset) = 0$, and let $B(F)$ be its base polytope. Given any $x \in \mathbb{R}^{\vert V \vert}$, one can compute
\[
\arg\max_{w \in B(F)} \langle x, w\rangle
\]
 using $O(\vert V \vert)$ calls to an evaluation oracle for $F$. Furthermore $w$ is integral.
\end{lem}

\section{Deferred Proofs}

\label{sec:Deferred-Proofs}

\subsection{Proof of Lemma~\ref{lem:cut2dual}}
\label{proof:lem_cut2dual}
We define the primal and dual optima of this problem, which will be useful for the proof.
\begin{defn}[Graph subproblem minimizers]
Let $\txx^*$ be the minimizer of
\begin{align}
\underset{\xx}{\min}\, g(\xx) + \phi(\xx)\,,
\end{align}
and $\tyy^*$ be the minimizer of
\begin{align}
\underset{\yy\in B(G)}{\min}\, \phi^*(-\yy)\,.
\end{align}
\end{defn}

The main tool that we will use for this proof will be the following two
structural statements, which can be
extracted from Propositions 4.2 and 8.3 in \cite{bach2011learning}.
\begin{lem}[\cite{bach2011learning}]
Consider any
submodular function $F:2^V\rightarrow\mathbb{R}$.
\begin{enumerate}
\item{
Fix some $\xx\in\mathbb{R}^n$.
For any $\yy\in\mathbb{R}^n$, $\yy$ 
is an optimizer of $\underset{\yy\in B(F)}{\max}\, \langle \yy, \xx\rangle$
if and only if there exists a permutation $\pi$ of $[n]$ such that
$x_{\pi_1}\geq x_{\pi_2}\geq \dots\geq x_{\pi_n}$
and for all $u\in V$ we have
\begin{align*}
w_u = \begin{cases}
F(\{\pi_1\}) & \text{if $u = 1$}\\
F(\{\pi_1,\pi_2,\dots,\pi_u\}) - F(\{\pi_1,\pi_2,\dots,\pi_{u-1}\}) & \text{if $u \geq 2$}\\
\end{cases}
\,.
\end{align*}
	}
\item{
Given a function $\phi$ that satisfies the conditions in Definition~\ref{defn:conditions},
the optimal solution $\txx^*$ to the problem
\begin{align*}
\underset{\xx}{\min}\, f(\xx) + \phi(\xx)\,,
\end{align*}
where $f$ is the Lov\'{a}sz extension of $F$,
is given by 
\begin{align*}
\tx_u^* = -\inf(\{\lambda\in\mathbb{R}\ :\ u\in S(\lambda)\})
\end{align*}
for all $u\in V$,
where 
\begin{align*}
S(\lambda) = \underset{S\subseteq V}{\mathrm{argmin}}\, F(S) + \sum\limits_{u\in S} \phi_u'(-\lambda)\,.
\end{align*}
}
\end{enumerate}
\label{lem:mincut_characterization}
\end{lem}
Additionally, we present two simple lemmas which will be useful in the proof. The first one
upper bounds the $\ell_1$ diameter of a base polytope, and the second one upper bounds
the $\ell_1$ norm of the gradient of a function in the base polytope.
\begin{lemma}
For any submodular function $F:2^V\rightarrow\mathbb{R}_{\geq 0}$
and $F(S) \leq \fmax$ for all $S\subseteq V$,
we have that 
\[ \underset{\yy\in B(F)}{\max}\, \|\yy\|_1 \leq 2n\fmax\,.\]
\label{lem:y_ub}
\end{lemma}
\begin{proof}
By definition of $B(F)$, for all $u\in V$ we have $w_u \leq F(\{u\}) \leq \fmax$, so
$\sum\limits_{u\in V:w_u\geq 0} w_u \leq n \fmax$.
Also, $\sum\limits_{u\in V} w_u = F(V)$, so we conclude that
\begin{align*}
\|\yy\|_1 
& = \sum\limits_{u\in V:w_u\geq 0} w_u - \sum\limits_{u\in V:w_u< 0} w_u  \\
& = 2 \sum\limits_{u\in V:w_u\geq 0} w_u - F(V) \\
& \leq 2 n \fmax - F(V) \\
& \leq 2 n \fmax 
\end{align*}
\end{proof}
\begin{lemma}
For any submodular function $F:2^V\rightarrow\mathbb{R}_{\geq 0}$, $F(S)\leq F_{\max}$ for all $S\subseteq V$,
and function $\psi:\mathbb{R}^n\rightarrow\mathbb{R}$ satisfying the conditions of Definition~\ref{defn:conditions}
we have that 
\begin{align*}
\underset{\yy\in B(F)}{\max}\, \|\nabla\psi^*(-\yy)\|_1 \leq \frac{2n\fmax}{\str} + \|\nabla\psi^*(\zerov)\|_1\,.
\end{align*}
\label{lem:gradient_ub}
\end{lemma}
\begin{proof}
\begin{align*}
 \|\nabla\psi^*(-\yy)\|_1 
& \leq \|\nabla\psi^*(-\yy) - \nabla \psi^*(\zerov)\|_1 + \|\nabla\psi^*(\zerov)\|_1\\
& \leq \frac{1}{\str} \|\yy\|_1 + \|\nabla\psi^*(\zerov)\|_1\\
& \leq \frac{2n \fmax}{\str} + \|\nabla\psi^*(\zerov)\|_1\,,
\end{align*}
where we used the triangle inequality, the $\frac{1}{\str}$-smoothness of the $\psi_u^*$'s, and Lemma~\ref{lem:y_ub}.
\end{proof}
We are now ready to proceed with the proof.
\begin{proof}[Proof of Lemma~\ref{lem:cut2dual}]
We let $\Lambda = \{\lambda_1, \dots, \lambda_k\}$, where $\lambda_1 < \dots < \lambda_k$, and define 
$S(\lambda)$ to be a minimal set in 
\begin{align*}
\underset{S\subseteq V}{\mathrm{argmin}}\, G(S) + \sum\limits_{u\in S} \phi_u'(-\lambda)
\end{align*}
for all $\lambda\in\mathbb{R}$.
Note that this can be equivalently written as 
\begin{align*}
& \underset{S\subseteq V}{\mathrm{argmin}}\, G(S) + \sum\limits_{u\in S} \max\{0,\phi_u'(-\lambda)\}
- \sum\limits_{u\in S} \max\{0,-\phi_u'(-\lambda)\} \\
& = \underset{S\subseteq V}{\mathrm{argmin}}\, G(S) + \sum\limits_{u\in S} \max\{0,\phi_u'(-\lambda)\}
+ \sum\limits_{u\in V\backslash S} \max\{0,-\phi_u'(-\lambda)\} 
- \sum\limits_{u\in V} \max\{0,-\phi_u'(-\lambda)\} \\
& = \underset{S\subseteq V}{\mathrm{argmin}}\, G(S) + \sum\limits_{u\in S} \max\{0,\phi_u'(-\lambda)\}
+ \sum\limits_{u\in V\backslash S} \max\{0,-\phi_u'(-\lambda)\} \\
& = \underset{S\subseteq V}{\mathrm{argmin}}\, c_{\lambda}^+(S\cup \{s\})\,,
\end{align*}
by the definition of the parametric capacities $c_{\lambda}$,
where $c_{\lambda}^+(S\cup\{s\}) = \sum\limits_{\substack{u\in S\cup\{s\}\\v\in V\backslash (S\cup\{s\})}} c_{\lambda}(u,v)$.
Additionally, we denote $\eps = \frac{1}{3 \smo}$ for convenience.
By the second item of Lemma~\ref{lem:mincut_characterization}, we know that 
the minimizer of 
$\underset{\xx}{\min}\, f(\xx) + \phi(\xx)$
is defined as
\begin{align*}
\tx_u^* = -\inf\{\lambda\in\mathbb{R}\ :\ u\in S(\lambda)\}
\,. 
\end{align*}
For all $u\in V$, let $i_u = \underset{i\in[k]}{\mathrm{argmin}}\, \{\lambda_i\ |\ u\in S(\lambda_i)\}$
and
$ \tx_u =
-\lambda_{i_u}\,.
 $
We will first prove that $\|\txx - \txx^*\|_\infty \leq \eps$.
Now, by definition we have that $\tx_u^* \geq \tx_u$. Additionally,
setting $\lambda_0 = -\infty$ for convenience,
we have $u\notin S(\lambda_{i_u-1})$, and $u\in S(-\tx_u^*)$, so 
$S(\lambda_{i_u-1}) \subset S(-\tx_u^*)$. 
By the first item of Definition~\ref{d:eps-pmstc}, this implies that 
\begin{align*}
-\tx_u^* \geq \lambda_{i_u} - \eps = -\tx_u - \eps \Leftrightarrow \tx_u^* \leq \tx_u + \eps\,.
\end{align*}
Therefore, we have concluded that $\left|\tx_u - \tx_u^*\right| \leq \eps$ for all $u\in V$, i.e. $\|\txx - \txx^*\|_\infty \leq \eps$.

We define a dual solution $\hyy = -\nabla \phi(\txx)$.
We will show that $\tyy^*$ can be retrieved by rounding $\hyy$.
Using the fact that $\phi_u$'s are $\smo$-smooth
and the optimality condition $\tyy^* = -\nabla\phi(\txx^*)$ from Lemma~\ref{lem:primal-dual-conversion-exact}, we get that
\begin{align*}
\|\hyy - \tyy^*\|_\infty 
 = \|\nabla\phi(\txx) - \nabla \phi(\txx^*)\|_\infty 
 \leq \smo \|\txx -\txx^*\|_\infty 
 \leq \smo \eps
= 1 / 3\,.
\end{align*}
On the other hand, by optimality of $\tyy^*$, 
it is a maximizer of $\underset{\yy\in B(G)}{\max}\, \langle \yy, \txx^*\rangle$.
By the first item of Lemma~\ref{lem:mincut_characterization},
there exists a permutation $\pi_1,\dots,\pi_n$ of $V$ such that
$\ty_u^* = G(\{\pi_1,\dots,\pi_u\}) - G(\{\pi_1,\dots,\pi_{u-1}\})$ for $u\in V$.
As $G$ takes integral values, %
we have $\ty_u^*\in \Delta \cdot \mathbb{Z}$ for all $u\in V$, and since $|\hy_u - \ty_u^*| < 1 / 2$, we can exactly recover $\tyy^*$ by rounding
each entry of $\hyy$ to the closest integer.

Our next goal is to compute a $G$-decomposition of $\tyy^*$, which we will do by computing an exact primal solution and then
again applying the first item of Lemma~\ref{lem:mincut_characterization}.
Given $\tyy^*$, we can easily recover the primal optimum $\txx^* = \nabla\phi(-\tyy^*)$.
In order to recover a decomposition $\tyy^* = \sum\limits_{i=1}^r \tyy^{*i}$, we use the well-known fact~\cite{edmonds2003submodular} that
\begin{align*}
\underset{\yy\in B(G)}{\max}\, \langle \yy, \xx\rangle
=
\underset{\yy^i\in B(G_i)}{\max}\, \sum\limits_{i=1}^r \langle \yy^i, \xx\rangle\,,
\end{align*}
so
for any $i\in[r]$, 
$\tyy^{*i}$ necessarily maximizes
\[ \underset{\yy^{i}\in B(G_i)}{\max}\, \langle \yy^{i}, \txx^*\rangle\,. \]
Therefore, by the first item of Lemma~\ref{lem:mincut_characterization}, $\tyy^{*i}$ can be recovered by sorting the entries of $\txx^*$ in decreasing order, such that
$\tx_{\pi_1}^* \geq \tx_{\pi_2}^*\geq \dots\geq \tx_{\pi_n}^*$
for some permutation $\pi$ of $V$, and then setting
\begin{align}
\ty_u^{*i} = %
G_i(\{\pi_1,\dots,\pi_u\}) - G_i(\{\pi_1,\dots,\pi_{u-1}\}) %
\label{eq:greedy_decomposition}
\end{align}
for all $u\in V$.
Note that $\tyy^{*i}$'s are in $\mathbb{Z}^n$.

The runtime is dominated by the computation of the decomposition in
(\ref{eq:greedy_decomposition}), which involves computing prefix cuts for each $G_i$
and by Lemma~\ref{lem:prefix_cuts}
takes time $O\left(\sum\limits_{i=1}^r |V_i|^2\right)$.
Therefore, the total runtime is $O\left(n + \sum\limits_{i=1}^r |V_i|^2\right)$. 
\end{proof}

\begin{lemma}[Computing all prefix cut values]
\label{lem:prefix_cuts}
Given a graph $G(V,E,c)$ with $V = \{1,2,\dots,n\}$, we can compute the values 
$c^+([u])$ for all $u\in[n]$ in time $O(n^2)$.
\end{lemma}
\begin{proof}
We note that $c^+(\emptyset) = 0$ and 
for any $u \geq 1$ we have 
\begin{align}
c^+([u]) = c^+([u-1]) + \sum\limits_{v=u+1}^n c_{uv} - \sum\limits_{v=1}^{u-1} c_{vu}\,. 
\label{eq:prefix_cut}
\end{align}
Therefore $c^+([u])$ can be computed in $O(n)$ given $c^+([u-1])$. As we apply (\ref{eq:prefix_cut}) $n$ times,
the total runtime is $O(n^2)$.
\end{proof}

\subsection{Proof of Lemma~\ref{lem:progress_lemma}}
\label{proof:lem_progress_lemma}

We first prove the following lemma, which helps us bound the range of parameters for parametric min $s,t$-cut.
\begin{lemma}
Consider a graph $G(V\cup\{s,t\},E,c\geq \zerov)$ %
and a function $\phi(\xx) = \sum\limits_{u\in V} \phi_u(x_u)$ that satisfies Assumption~\ref{defn:conditions}.
Additionally, let $G(S) = c^+(S\cup \{s\})$ for all $S\subseteq V$ be the cut function associated with the graph.
For any $\lambda\in\mathbb{R}$, we set 
$S(\lambda)$ to be the smallest set that minimizes
\[ \underset{S\subseteq V}{\mathrm{min}}\, G(S) + \sum\limits_{u\in S} \phi_u'(-\lambda)\,. \]
Let $\rho = \underset{u\in V}{\max}\, |\phi_u'(0)|$ and 
$G_{\max} = \underset{S\subseteq V}{\max}\, G(S)$.
Then,
$S(\lambda_{\min}) = \emptyset$ and $S(\lambda_{\max}) = V$, where
$\lambda_{\min} = -2\frac{\rho + G_{\max}}{\str}$ and $\lambda_{\max} = 2\frac{\rho + G_{\max}}{\str}$.
\label{lem:lambda_range}
\end{lemma}
\begin{proof}
We first note that by the $\str$-strong convexity of the $\phi_u$'s,
and since $-\lambda_{\min} > 0 > -\lambda_{\max}$, we have that
\[ \phi_u'(-\lambda_{\min}) \geq \phi_u'(0) + \str |\lambda_{\min}| \,.\]
and
\[ \phi_u'(-\lambda_{\max}) \leq \phi_u'(0) - \str |\lambda_{\max}| \,.\]
Therefore for any $\emptyset \neq S\subseteq V$ we have
\begin{align*}
 G(S) + \sum\limits_{u\in S} \phi_u'(-\lambda_{\min})
& \geq G(S) + \sum\limits_{u\in S} (\phi_u'(0) + \str |\lambda_{\min}|)\\
& \geq G(S) - |S| \rho + |S| \str |\lambda_{\min}|\\
& = G(S) - |S| \rho + 2|S| \str \frac{\rho + G_{\max}}{\str}\\
& > G(S) + G_{\max}\\
& \geq G(\emptyset)\,,
\end{align*}
where we used the fact that $G(S) \geq 0$ and $G(\emptyset) \leq G_{\max}$,
so $S(\lambda_{\min}) = \emptyset$. Similarly, for any $S\subset V$ we have
\begin{align*}
 G(S) + \sum\limits_{u\in S} \phi_u'(-\lambda_{\max})
& =  G(S) + \sum\limits_{u\in V} \phi_u'(-\lambda_{\max}) - \sum_{u\in V\backslash S} \phi_u'(-\lambda_{\max})\\
& \geq  G(S) + \sum\limits_{u\in V} \phi_u'(-\lambda_{\max}) - \sum_{u\in V\backslash S} \left(\phi_u'(0) - \str |\lambda_{\max}|\right)\\
& \geq  G(S) + \sum\limits_{u\in V} \phi_u'(-\lambda_{\max}) + |V\backslash S| \left(\str |\lambda_{\max}| - \rho\right)\\
& =  G(S) + \sum\limits_{u\in V} \phi_u'(-\lambda_{\max}) + |V\backslash S| \left(2\str \frac{\rho + G_{\max}}{\str} - \rho\right)\\
& > G(S) + \sum\limits_{u\in V} \phi_u'(-\lambda_{\max}) + G_{\max}\\
& \geq G(V) + \sum\limits_{u\in V} \phi_u'(-\lambda_{\max}) \,,
\end{align*}
where we used the fact that $G(S) \geq 0$ and $G(V) \leq G_{\max}$,
so $S(\lambda_{\max}) = V$.
\end{proof}

We are now ready for the proof.
\begin{proof}[Proof of Lemma~\ref{lem:progress_lemma}]
We first shift the polytope $B(F)$ so that $\yy$ is translated to $\zerov$.
Specifically, for all $S\subseteq V$,
we let $\hF(S) = F(S) - \yy(S)$ 
and $\hF_i(S) = F_i(S) - \yy^i(S)$ for all $i\in[r]$.
As we are just subtracting a linear function, $\hF$ and the $\hF_i$'s are still submodular functions,
and $B(\hF) = B(F) - \yy$, $B(\hF_i) = B(F_i) - \yy^i$ for all $i\in[r]$.
Note that $\yy^i\in B(F_i)$ implies that the $\hF_i$'s (and thus also $\hF$) are non-negative, since
\[ \hF_i(S) = F_i(S) - \yy(S) \geq 0\,, \]
and additionally 
$\hF_i(\emptyset) = F_i(\emptyset) = 0$ and
$\hF_i(V_i) = F_i(V_i) - \yy^i(V_i) = 0$ for all $i\in[r]$
(also implying $\hF(\emptyset) = \hF(V) = 0$).

We run the algorithm from Lemma~\ref{lem:approx-cut} on the 
$\hF_i$'s to obtain directed graphs $G_i(V,E,c^i\geq 0)$
whose
($V_i$-restricted)
cut functions $G_i(S) = c^{i+}(S)$ $\C$-approximate $\hF_i(S)$, where $\C=\underset{i\in[r]}{\max}\, \{|V_i|^2/4+|V_i|\}$.
More specifically,
\begin{align}
\frac{1}{\C} \hF_i(S) \leq G_i(S) \leq \hF_i(S) \text{ for all $S\subseteq V_i$}, \ G_i(V_i) = \hF_i(V_i)
\end{align}
and
\begin{align}
\frac{1}{\C} \left(B(F_i) - \yy^i\right) = \frac{1}{\C} B(\hF_i) \subseteq B(G_i) \subseteq B(\hF_i) = B(F_i) - \yy^i\,,
\label{eq:polytope_approx}
\end{align}
We also define the graph $G(V,E,c\geq 0)$, where $c = \sum\limits_{i=1}^r c^i$ and has cut function
$G(S) = \sum\limits_{i=1}^r G_i(S)$ for all $S\subseteq V$. Then, 
\begin{align}
\frac{1}{\C} \hF(S) \leq G(S) \leq \hF(S) \text{ for all $S\subseteq V$}, \ G(V) = \hF(V)
\end{align}
and
\begin{align}
\frac{1}{\C} \left(B(F) - \yy\right) \subseteq B(G) \subseteq B(F) - \yy\,.
\label{eq:polytope_approx_sum}
\end{align}

We absorb the linear term that we subtracted from $F$ into
the parametric function. Concretely, we define $\phi$ as 
$\phi(\xx) = \psi(\xx) + \langle \yy, \xx\rangle$ for all $\xx\in\mathbb{R}^n$.
It is easy to see that $\phi(\xx)$ is coordinate-wise separable, as
$\phi(\xx) = \sum\limits_{u\in V}\phi_u(x_u)$
where
$\phi_u(x_u) = \psi_u(x_u) + w_u x_u$ for all $u\in V$
and that it satisfies Assumption~\ref{defn:conditions}, since
$\phi_u''(x_u) = \psi_u''(x_u)$ and 
\[ |\phi_u'(0)| = |\psi_u'(0) + w_u| \leq |\psi_u'(0)| + F(\{u\}) \leq |\psi_u'(0)| + \fmax = n^{O(1)}\,. \]
Additionally, the Fenchel dual of $\phi_u$ is a shifted version of $\psi_u$, i.e.
\[ \phi_u^*(z) = \underset{y\in\mathbb{R}}{\max}\, zy - \phi(y) = 
\underset{w\in\mathbb{R}}{\max}\, zy - \psi(y) - w_u y = 
\underset{w\in\mathbb{R}}{\max}\, (z-w_u) y - \psi(y) = 
\psi_u^*(z-w_u)\,.
\]
We will now run the algorithm from Lemma~\ref{lem:cut2dual} on graphs $G_i$
and parametric function $\phi$ to obtain a dual solution vector $\tyy$.
We note that, as $F_i$'s and $\yy^i$'s take integer values, $G_i$'s take integer values too.
This algorithm takes a $\frac{1}{3\smo}$-approximate parametric min $s,t$-cut as input, which we first compute using
Theorem~\ref{t:aprx-pmc}, with range of parameters $[\lambda_{\min}, \lambda_{\max}]$ given by Lemma~\ref{lem:lambda_range}. 
We have 
\begin{align*}
\lambda_{\max} - \lambda_{\min} 
& = 
4 \frac{\underset{u\in V}{\max}\, |\phi_u'(0)| + \underset{S\subseteq V}{\max}\, G(S)}{\str}\\
&\leq 4 \frac{\underset{u\in V}{\max}\, |\phi_u'(0)| + \hF_{\max}}{\str}\\
&\leq 4 \frac{\underset{u\in V}{\max}\, |\phi_u'(0)| + F_{\max} + \|\yy\|_1}{\str}\\
&\leq 4 \frac{\underset{u\in V}{\max}\, |\phi_u'(0)| + (2n+1)\fmax}{\str}\\
& = n^{O(1)}\,,
\end{align*}
where we used Lemma~\ref{lem:y_ub}
and the fact that the
quantities $|\phi_u'(0)|, \fmax, \frac{1}{\str}$ are bounded by $n^{O(1)}$ (Assumption~\ref{defn:conditions}).

Based on Theorem~\ref{t:aprx-pmc}, the time to obtain the $\frac{1}{3\smo}$-approximate parametric min $s,t$-cut
will be 
\[ O\left(\tmf(n, |E'|\log n) \log \frac{\lambda_{\max}-\lambda_{\min}}{1 / 3\smo}\log n\right) = 
\tO{\tmf\left(n, n + \sum\limits_{i=1}^r |V_i|^2\right)}
\,. \]

So, by applying Lemma~\ref{lem:cut2dual}, we obtained a dual solution $\tyy = \sum\limits_{i=1}^r \tyy^i$
for which
$\tyy^i \in B(G_i)$ and 
\begin{align}
\tyy 
= \underset{\tyy\in B(G)}{\mathrm{argmin}}\, \phi^*(-\tyy)
= \underset{\tyy\in B(G)}{\mathrm{argmin}}\, \psi^*(-\yy-\tyy)\,.
\label{eq:inner_guarantee}
\end{align}
\ifx 0
\begin{align}
\phi^*(-\tyy) & \leq \underset{\tyy^*\in B(G)}{\min}\, \phi^*(-\tyy^*) + 
\underset{\zz\in B(G)}{\max}\, \|\nabla\phi^*(-\zz)\|_1\cdot
\Delta r + \frac{n\Delta^2 r^2}{2\str}\\
\Rightarrow \psi^*(-\yy -\tyy) & \leq \underset{\tyy^*\in B(G)}{\min}\, \psi^*(-\yy-\tyy^*) + \underset{\zz\in B(G)}{\max}\, \|\nabla\psi^*(-\yy-\zz)\|_1 \cdot \Delta r + \frac{n\Delta^2 r^2}{2\str}\notag\\
& \leq \underset{\tyy^*\in B(G)}{\min}\, \psi^*(-\yy-\tyy^*) + \underset{\zz\in B(F)}{\max}\, \|\nabla\psi^*(-\zz)\|_1
\cdot \Delta r + \frac{n\Delta^2 r^2}{2\str}
\notag\\
& \leq \underset{\tyy^*\in B(G)}{\min}\, \psi^*(-\yy-\tyy^*) + \left(\frac{2n\fmax}{\str} + \|\nabla\psi^*(\zerov)\|_1\right)
\Delta r + \frac{n\Delta^2 r^2}{2\str}\notag\\
& \leq \underset{\tyy^*\in B(G)}{\min}\, \psi^*(-\yy-\tyy^*) + \left(\frac{\fmax}{\str} + \frac{\|\nabla\psi^*(\zerov)\|_1}{2 n}\right)
\eps
+ \frac{\eps^2}{8\str n}
\,,
\label{eq:inner_guarantee}
\end{align}
\fi

For all $u\in V$ and $i\in[r]$, we 
set $w_u'^i = w_u^i + \ty_u^i$
and $w_u' = \sum\limits_{i=1}^r w_u'^i$. %
Note that these quantities are still integral.
We now prove the two parts of the lemma statement (feasibility
and optimality) separately.

\noindent {\bf Feasibility. } 
For any $i\in [r]$ %
and $S\subseteq V_i$, 
we have that 
\begin{align*}
\yy'^i(S) 
 = \yy^i(S) + \tyy^i(S) %
 \leq \yy^i(S) + G_i(S) 
 \leq \yy^i(S) + \hF_i(S) 
 = \yy^i(S) + F_i(S) - \yy^i(S) 
 = F^i(S)\,,
\end{align*}
where the second inequality follows from %
the fact that $\tyy^{*i}\in G_i$.
Similarly, we have 
\begin{align*}
\yy'^i(V_i) 
& = \yy^i(V_i) + \tyy^i(V_i) 
 = \yy^i(V_i) + G_i(V_i) 
 = \yy^i(V_i) + \hF_i(V_i) 
 = F_i(V_i) \,.
\end{align*}
So we conclude that $\yy'^i\in B(F_i)$ for all $i\in[r]$. %

\noindent {\bf Optimality. }
Let's set $h(z) := \psi^*(-z)$ for all $z\in\mathbb{R}^n$ for notational convenience, so that our goal is to prove that
\begin{align*}
h(\yy') - h(\yy^*) \leq \left(1 - \frac{1}{\C}\right) \left(h(\yy) - h(\yy^*)\right) \,. %
\end{align*}
We will prove a slightly different statement
where $\yy'$ is replaced by a solution on the path from $\yy$ to $\yy^*$,
which is enough because $\yy'$ is 
optimal in $\yy + B(G)$.
Concretely, 
we set $\oyy = \frac{1}{\C} \left(\yy^*-\yy\right)$ and instead will prove 
\begin{align*}
h(\yy + \oyy) - h(\yy^*) \leq \left(1 - \frac{1}{\C}\right) \left(h(\yy) - h(\yy^*)\right)\,.
\end{align*}
Now, $h$ is a convex function, so applying convexity twice we have
\begin{equation}
\begin{aligned}
h(\yy) & \geq h(\yy + \oyy) + \langle \nabla h(\yy + \oyy), - \oyy\rangle \\
& = h(\yy + \oyy) - \frac{1}{\C} \langle \nabla h(\yy + \oyy), \yy^* - \yy\rangle
\label{eq:convex1}
\end{aligned}
\end{equation}
and
\begin{equation}
\begin{aligned}
h(\yy^*) 
&\geq h(\yy + \oyy) + \langle \nabla h(\yy +  \oyy), \yy^* - \yy - \oyy\rangle\\
& = h(\yy + \oyy) + \frac{\C-1}{\C}\langle \nabla h(\yy +  \oyy), \yy^* - \yy\rangle\,.
\label{eq:convex2}
\end{aligned}
\end{equation}
We divide (\ref{eq:convex2}) by $\C-1$ and then sum it with (\ref{eq:convex1}), getting
\begin{align}
& h(\yy) \notag
+ \frac{1}{\C-1} h(\yy^*) 
\geq h(\yy + \oyy)+ \frac{1}{\C-1} h(\yy + \oyy)\,.
\end{align}
Equivalently,
\begin{align}
\frac{\C}{\C-1} \left(h(\yy + \oyy) - h(\yy^*)\right)\notag
\leq 
 h(\yy) 
- h(\yy^*) \,.
\end{align}
So by rearranging,
\begin{align}
  h(\yy + \oyy) - h(\yy^*)
\leq \left(1 - \frac{1}{\C}\right)\left(
 h(\yy) 
- h(\yy^*) \right)\,.
\end{align}
Thus we can equivalently write that
\begin{align}
\psi^*(-\yy - \oyy) - \psi^*(-\yy^*)
\leq \left(1 - \frac{1}{\C}\right)\left(
 \psi^*(-\yy) 
- \psi^*(-\yy^*) \right)\,.
\label{eq:convex12}
\end{align}
Now, since by (\ref{eq:polytope_approx_sum}) we have $\frac{1}{\C} (B(F) - \yy) \subseteq B(G)$ and $\yy^* \in B(F)$,
we have $\oyy = \frac{1}{\C} (\yy^* - \yy) \in B(G)$. Combinging the fact that $\tyy$ is a %
minimizer of $\underset{\tyy^*\in B(G)}{\min}\, \psi^*(-\yy-\tyy^*)$ with
(\ref{eq:inner_guarantee}) and the fact that $\oyy \in B(G)$, we have
\begin{align*}
\psi^*(-\yy')
= \psi^*(-\yy -\tyy)
\leq
\psi^*(-\yy -\oyy) \,. %
\end{align*}
Combining this with (\ref{eq:convex12}), we obtain the desired claim:
\begin{align}
& \psi^*(-\yy') - \psi^*(-\yy^*)
\leq \left(1 - \frac{1}{\C}\right)\left(
 \psi^*(-\yy) 
- \psi^*(-\yy^*) \right)\,.
\label{eq:progress_intermediate}
\end{align}
The running time to compute graphs $G_i$ is $O\left(\sum\limits_{i=1}^r |V_i|^2 \opto_i\right)$ 
and the time to run the algorithm from Lemma~\ref{lem:cut2dual} is $\tO{n + \sum\limits_{i=1}^r |V_i|^2}$,
so the total running time is 
\[ \tO{\sum\limits_{i=1}^r |V_i|^2 \opto_i + \tmf\left(n, n + \sum\limits_{i=1}^r |V_i|^2\right) }\,. \]
\end{proof}

\begin{algorithm}[h!]
\caption{Finding all minimum cuts}
\begin{algorithmic}[1]
\vspace{0.2cm}
\STATE {\bfseries function }{\textsc{\textsc{FindMinCuts}$(G(V,E,c),\phi, \eps)$}}
\Indent
\STATE $V'=V\cup\{s,t\}$, $E' = E\cup\underset{u\in V}{\bigcup} \{(u,t)\}\cup\underset{u\in V}{\bigcup} \{(s,u)\}$
\STATE {Define parametric capacities
\begin{align*}
c_\lambda(u,v) = \begin{cases}
	\max\{0,\phi_u'(-\lambda)\} & \text{if $u\in V, v=t$}\\
	\max\{0,-\phi_u'(-\lambda)\} & \text{if $u=s, v\in V$}\\
c_{uv} & \text{otherwise}
\end{cases}
\end{align*}
}
\STATE Set $ (\Lambda,\tau) = $\textsc{ApxParametricMinCut}$(G'(V',E'),c_{\lambda}, \lambda_{\min} = -n^{O(1)}, \lambda_{\max} = n^{O(1)})$ %
	\STATE Set $\ty_u = -\phi_u'^*(-\tau(u))$ for all $u\in V$ %
\RETURN $\tyy$
\EndIndent
\end{algorithmic}
\label{alg:findmincuts}
\end{algorithm}

\bibliographystyle{abbrv}
\bibliography{main}

\end{document}